\documentclass[conference]{IEEEtran}
\IEEEoverridecommandlockouts

\usepackage{cite}
\usepackage{amsmath,amssymb,amsfonts}
\usepackage{graphicx}
\usepackage{textcomp}
\usepackage{xcolor}
\def\BibTeX{{\rm B\kern-.05em{\sc i\kern-.025em b}\kern-.08em
    T\kern-.1667em\lower.7ex\hbox{E}\kern-.125emX}}

\usepackage{algorithm}
\usepackage{algpseudocode}
\usepackage{booktabs}
\usepackage{multirow}
\usepackage{braket}
\usepackage{pgfplots}
\usepackage{hyperref}
\usepackage{amsthm} 

\makeatletter
\algrenewcommand\algorithmiccomment[1]{\hskip0pt// #1}
\makeatother

\newtheorem{theorem}{Theorem}[section]
\theoremstyle{definition}
\newtheorem{definition}[theorem]{Definition}

\pgfplotsset{compat=1.17}
\usetikzlibrary{patterns, fillbetween}
\usepackage{pifont}%
\newcommand{\cmark}{\ding{51}}%
\newcommand{\xmark}{\ding{55}}%

\newcommand{\thm}[1]{\hyperref[thm:#1]{Theorem~\ref*{thm:#1}}}
\newcommand{\alg}[1]{\hyperref[alg:#1]{Algorithm~\ref*{alg:#1}}}
\renewcommand{\sec}[1]{\hyperref[sec:#1]{Section~\ref*{sec:#1}}}
\newcommand{\fig}[1]{\hyperref[fig:#1]{Figure~\ref*{fig:#1}}}
\newcommand{\tab}[1]{\hyperref[tab:#1]{Table~\ref*{tab:#1}}}
\newcommand{\df}[1]{\hyperref[def:#1]{Definition~\ref*{def:#1}}}

\DeclareMathOperator{\photon}{photon}
\DeclareMathOperator{\fusee}{fusee}
\DeclareMathOperator{\measuree}{measuree}
\DeclareMathOperator{\abs}{abs}
\DeclareMathOperator{\NetIndex}{LayerIndex}
\DeclareMathOperator{\List}{List}

\DeclareMathOperator{\Parent}{Parent}
\DeclareMathOperator{\maxparent}{MTime}

\DeclareMathOperator{\best}{best}
\DeclareMathOperator{\local}{local}
\DeclareMathOperator{\remote}{remote}
\DeclareMathOperator{\GBP}{GBP}
\DeclareMathOperator{\LSP}{LSP}
\DeclareMathOperator{\new}{new}
\DeclareMathOperator{\current}{current}
\DeclareMathOperator{\init}{init}

\definecolor{lightblue}{rgb}{0.63, 0.74, 0.78}
\definecolor{seagreen}{rgb}{0.18, 0.42, 0.41}
\definecolor{orange}{rgb}{0.85, 0.55, 0.13}
\definecolor{silver}{rgb}{0.69, 0.67, 0.66}
\definecolor{rust}{rgb}{0.72, 0.26, 0.06}
\definecolor{customgray}{rgb}{0.69921875, 0.69921875, 0.69921875}
\definecolor{customgray2}{RGB}{217,217,217}
\definecolor{customred}{rgb}{0.80859375, 0.1484375, 0.15234375}
\definecolor{customred1}{RGB}{197, 58, 50}
\definecolor{customred2}{RGB}{229, 115, 115}
\definecolor{customred3}{RGB}{239, 154, 154}
\definecolor{customred4}{RGB}{255, 205, 210}
\definecolor{customblue}{rgb}{0.12109375, 0.46484375, 0.703125}
\definecolor{customblue1}{RGB}{58,94,158}
\definecolor{customblue2}{RGB}{102,144,190}
\definecolor{customblue3}{RGB}{153,186,213}

\newcommand{\name}{DC-MBQC}

\begin{document}

\title{\name: A Distributed Compilation Framework for Measurement-Based Quantum Computing
}

\author{
    \IEEEauthorblockN{
        Yecheng Xue\textsuperscript{1}, 
        Rui Yang\textsuperscript{1}, 
        Zhiding Liang\textsuperscript{2}, and 
        Tongyang Li\textsuperscript{1}\IEEEauthorrefmark{1}
    }
    
    \vspace{1ex}
    
    \IEEEauthorblockA{
        \textsuperscript{1}School of Computer Science, Peking University, Beijing, China \\
        \textsuperscript{2}Department of Computer Science and Engineering, The Chinese University of Hong Kong, Hong Kong, China \\
        Email: \{mycts, ypyangrui, tongyangli\}@pku.edu.cn, zliang@cse.cuhk.edu.hk
    }
    \thanks{\IEEEauthorrefmark{1}Corresponding author.} 
}

\maketitle

\begin{abstract}
Distributed quantum computing (DQC) is a promising technique for scaling up quantum systems. 
While significant progress has been made in DQC for quantum circuit models, there exists much less research on DQC for measurement-based quantum computing (MBQC), which is a universal quantum computing model that is essentially different from the circuit model and particularly well-suited to photonic quantum platforms. 
In this paper, we propose \name, \textit{the first distributed quantum compilation framework} tailored for MBQC. 
We identify and address two key challenges in enabling DQC for MBQC. 
First, for task allocation among quantum processing units (QPUs), we develop an adaptive graph partitioning algorithm that preserves the structure of the graph state while balancing the workload across QPUs. 
Second, for inter-QPU communication, we introduce the layer scheduling problem and propose an algorithm to solve it. 
Regrading realistic hardware requirements, we optimize the execution time of running quantum programs and the corresponding required photon lifetime to avoid fatal failures caused by photon loss. 
Our experiments demonstrate a $7.46\times$ improvement on required photon lifetime and $6.82\times$ speedup with 8 fully-connected QPUs, which further confirm the advantage of distributed quantum computing in photonic systems. 

The source code is publicly available at \url{https://github.com/qfcwj/DC-MBQC}.
\end{abstract}

\section{Introduction} \label{sec:intro}

In recent years, quantum computing has gained enormous attention due to its potential speedup to solve problems including factorization~\cite{shor1999polynomial}, unstructured database search~\cite{grover1996fast}, and many other problems of practical interest~\cite{dalzell2023quantum}.
As the community moves forward from noisy-intermediate scale quantum era to early fault-tolerance~\cite{preskill2025beyond}, more developed quantum hardware appears with lower error rates in several technology stacks, including superconducting~\cite{ai2024quantum,gao2025establishing}, trapped ion systems~\cite{moses2023race,loschnauer2024scalable}, neutral atom systems~\cite{bluvstein2024logical}, and photonic chips~\cite{psiquantum2025manufacturable, aghaee2025scaling}.
At the current stage, it remains unclear which physical platform could be the final winner of fault-tolerant quantum computing (FTQC) where thousands of logical qubits simultaneously run both algorithmic logic and quantum error correction protocols on top of millions of physical qubits~\cite{gidney2021factor}. 
All of the above platforms would eventually require distributed quantum computing (DQC)~\cite{caleffi2024distributed} to scale up to fault-tolerant regime, where modularized quantum processing units (QPUs) are interconnected with optical fibers~\cite{main2025distributed}.

\begin{table}[t]
  \caption{Survey of distributed entangling generation (without distillation). $*$ means that the fidelity is estimated with post-selection results and can be overestimated. 
   }
  \label{tab:remote-entanglement}
  \resizebox{\columnwidth}{!}{
  \begin{tabular}{c|c|c|c}
    \toprule
    Platform  & Fidelity & Clock speed & Exp. \\
    \midrule
    \midrule
    Superconducting~\cite{leung2019deterministic}  & $79.3\%$ & $\sim$ MHz & \cmark \\
    Quantum dot~\cite{stockill2017phase}& $61.6\%$ & 7.3~kHz & \cmark \\
    \midrule
    Trapped ion~\cite{main2025distributed} & $86.1\%$ & 9.7~Hz & \cmark \\
    Trapped ion~\cite{stephenson2020high} & $94.0\%$ & 182~Hz & \cmark \\
    \midrule
    Neutral atom~\cite{ritter2012elementary}  & $98.7\%$* & 30~Hz & \cmark \\
    Neutral atom~\cite{li2024high}    & $\sim 99.9\%$ & $\sim$ 100~kHz & \textcolor{red}{\xmark} \\
    \midrule
    Photonic  & $99.72\%$*~\cite{psiquantum2025manufacturable} & $\sim$ MHz~\cite{aghaee2025scaling} &  \textcolor{green}{\cmark} \\
    \bottomrule
\end{tabular}
}
\end{table}

There are two important factors for the scalability of DQC: the fidelity of remote entanglement generation and the corresponding clock speed. 
Remote Bell pairs, combined with local unitary operations, are the essential mechanism for enabling arbitrary remote two-qubit gates (including CNOT and SWAP) in DQC. However, this process is inherently error-prone, typically introducing higher effective error rates than local operations.
According to \cite{sinclair2024fault}, it is estimated to require a remote entanglement fidelity above $90\%$ and a clock speed of MHz level to maintain the efficacy of quantum error correction (QEC) protocols in DQC.
By reaching these two thresholds, the error originated from noisy remote entangling gates can be suppressed as the distance of QEC codes increases and the background decoherence errors remain negligible compared to other noise sources in a fast QEC cycle.

We summarize recent progress of demonstrating this capability of generating remote entanglement between two QPUs in \tab{remote-entanglement}. 
Solid-state quantum hardware such as superconducting qubits and quantum dots typically have higher system clock rate but lower remote gate fidelity compared to atom-based systems.
Trapped ion systems have demonstrated the highest remote entanglement generation fidelity to date, which is $94\%$ without requiring post-selection~\cite{stephenson2020high}.
Ref.~\cite{main2025distributed} recently demonstrated a remote CZ gate with $86.1\%$ fidelity but with only 9.7 Hz of clock speed due to extra cooling requirement.
Neutral atom systems are becoming more popular in recent years because of their potential to run QEC experiments on a single QPU.
And there are several promising proposals~\cite{young2022architecture,covey2023quantum,li2024high} to improve the DQC capability for reconfigurable atom arrays, especially improving the remote entanglement generation rate from $10$~Hz~\cite{ritter2012elementary} to $10^5$~Hz level through atom multiplexing. 
Finally, photonic systems seem to be the most suitable candidate for DQC due to high remote entanglement fidelity ($99.72\%$) and high clock speed ($\sim$ MHz), demonstrated by the Omega chipset from PsiQuantum~\cite{psiquantum2025manufacturable} and Aurora system from Xanadu~\cite{aghaee2025scaling}.

Despite the potential advantages of implementing DQC, the associated compiler design for photonic systems is underexplored in the computer architecture community, even in the single QPU setting. 
OneQ~\cite{zhang2023oneq} is the first scalable compilation framework for measurement-based quantum computing (MBQC) on photonic hardware~\cite{raussendorf2003measurement,briegel2009measurement}, where the authors described multiple steps to map a quantum program in the circuit model to graph-based computation on photonic resource states and optimize the number of error-prone fusion operations~\cite{bartolucci2023fusion}.
Conventionally, quantum programs are designed using circuit models due to its simplicity and similarity to classical logical circuits.

In MBQC, more specifically one-way quantum computing (1WQC)~\cite{raussendorf2001one}, the computations are usually implemented in three steps. 
First, we prepare a large entangled state called the {\it resource state}, which is usually defined on a graph. 
In practice, it is experimentally hard to prepare and maintain such large-scale graph states~\cite{hein2004multiparty} with high enough fidelity to carry out meaningful computations.
The more scalable approach is to prepare a collection of small resource states and dynamically combine them into a larger resource state through fusion operations~\cite{bartolucci2023fusion}. %
Second, we implement a sequence of single qubit projective measurements known as the {\it measurement patterns}, which may depend on previous measurement results. 
Finally, a {\it byproduct correction} is required to address the inherent randomness of quantum measurement outcomes, as these random results cause the evolution of the unmeasured state to be non-deterministic and potentially subject to additional undesired evolutions.
We introduce the details of MBQC and fusion in \sec{mbqc} and \sec{mbqc-assumptions}.

Based on the results of OneQ, Ref.~\cite{mo2024fcm} further adopted circuit cutting technique to decompose a large circuit into several small sub-circuits and reduce the number of fusion operations in each sub-circuit.
However, this approach could face an exponential overhead when post-processing measurement statistics from each sub-circuit~\cite{jing2025circuit}.
More recently, OnePerc~\cite{zhang2024oneperc} observed that the probabilistic nature of fusion operations could lead to dramatic change to the computation graph in MBQC when fusion fails.
These efforts all lead to a more efficient and practical compiler design for MBQC photonic systems.

However, these previous works posed strong assumptions on the photonic hardware, especially the clock cycle of generating resource states at GHz level (1~ns/cycle) and photon storage time in the fibre-optical {\it delay line}~\cite{zhang2024oneperc, bombin2021interleaving}.
By overestimating these hardware requirements, the probability of losing a photon when executing long range operations could be comparable to the failure rate of error-prone fusion operations and lead to fatal failures when executing large-scale quantum programs on photonic systems.
Previous literature~\cite{bombin2021interleaving, zhang2023oneq} including OneQ estimates that the maximum photon lifetime stored in the delay line is about 5000 cycles with $\sim5\%$ photon loss probability.
According to \fig{required-lifetime}, this probability can be suddenly as high as $36.9\%$ and $99.9\%$ if the corresponding resource state generation rate becomes 10~ns/cycle and 100~ns/cycle, even higher than the current experimental fusion failure rate of $29\%$~\cite{guo2024boosted}.
Although it is possible to achieve the 1~ns/cycle clock rate in certain experimental settings~\cite{paesani2019generation,eltes2020integrated}, there is not yet clear evidence whether commercial photonic quantum hardware~\cite{psiquantum2025manufacturable,aghaee2025scaling} can achieve the same speed for generating resource states.
This imposes a fundamental restriction on the lifetime for which a photon can be stored in the delay line.

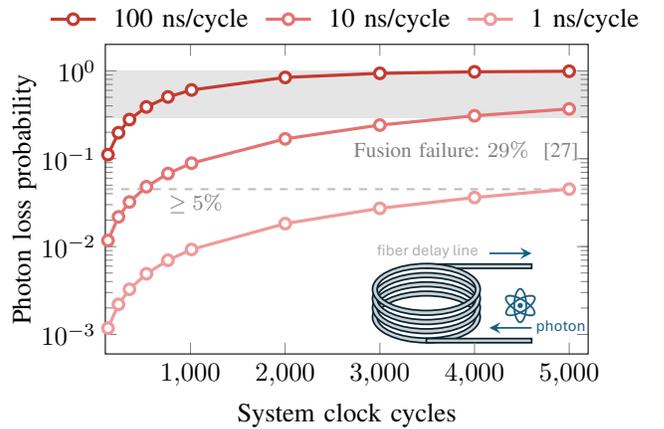
\begin{figure}[ht]
\centering
\begin{tikzpicture}
\node at (2.0,2.0) {
\includegraphics[width=0.16\textwidth]{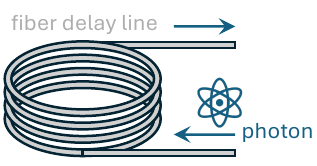}
};
\node at (0,3) {
\begin{tikzpicture}
    \begin{axis}[
        xlabel={System clock cycles},
        ylabel={Photon loss probability},
        ylabel style={yshift=-5pt},
        xmin=100,
        xmax=5200,
        ymin=0.6e-3, %
        ymax=2e0, %
        ymode=log,
        ytick={1e-3, 1e-2, 1e-1, 1e0},
        width=8cm,
        height=5.7cm,
        max space between ticks=40,
        legend style={draw=none, fill=none, at={(0.5,1.01)}, anchor=south, legend columns=3, column sep=4pt },
        clip mode=individual,
        legend cell align={left},
    ]
    \fill[gray, opacity=0.2] (axis cs:100,0.29) rectangle (axis cs:5200,1);
    \draw[dashed, gray, thick, opacity=0.5] (axis cs:100,0.045) -- (axis cs:5200,0.045);
    \addplot[
        very thick,
        mark=*,
        mark size=2pt, %
        mark options={fill=white, draw=customred1},
        color=customred1,
        opacity=1
    ] table [
        x=depth,
        y=probability,
        col sep=comma,
        header=true
    ] {data/lifetime/cycle-100ns.txt};
    \addlegendentry{100~ns/cycle};

    \addplot[
        very thick,
        mark=*,
        mark size=2pt, %
        mark options={fill=white, draw=customred2},
        color=customred2,
        opacity=1
    ] table [
        x=depth,
        y=probability,
        col sep=comma,
        header=true
    ] {data/lifetime/cycle-10ns.txt};
    \addlegendentry{10~ns/cycle};

    \addplot[
        very thick,
        mark=*,
        mark size=2pt, %
        mark options={fill=white, draw=customred3},
        color=customred3,
        opacity=1
    ] table [
        x=depth,
        y=probability,
        col sep=comma,
        header=true
    ] {data/lifetime/cycle-1ns.txt};
    \addlegendentry{ 1~ns/cycle};

    \end{axis}
\node at (4.8,2.7) {\footnotesize \textcolor{gray}{Fusion failure: $29\%$ ~\cite{guo2024boosted}} };
\node at (1.2,2.0) {\footnotesize \textcolor{gray}{$\geq5\%$} };
\end{tikzpicture}
};
\end{tikzpicture}
\caption{Demonstration of the effect of different resource state clock rates on photon loss probability.
The probability is estimated as $1-e^{-\alpha L}$, where $\alpha=0.2$ dB/km is the attenuation rate in state-of-the-art optical fibers and $L$ is the distance a photon traveled with $2/3$ speed of light in delay lines.
}
\label{fig:required-lifetime}
\end{figure}

Recent efforts have begun to address this critical physical limitation. 
For instance, OneAdapt~\cite{zhang2025oneadaptadaptivecompilationresourceconstrained} identified that a primary cause for long photon storage times arises from the need to route connections between distant computational steps. 
In the compilation model of its predecessor, OnePerc~\cite{zhang2024oneperc}, this could require storing a photon in a delay line for an unbounded number of clock cycles. 
To mitigate this, OneAdapt proposed a ``dynamic refresh'' mechanism. This technique actively manages the storage duration of each photon, refreshing those about to exceed a predefined lifetime limit by remapping them in the computation. 
This bounds the storage time required for such long-range connections.
However, OneAdapt focuses primarily on the storage time induced by routing, while a more comprehensive framework to systematically account for all sources of the required photon lifetime is still needed.

\begin{figure}[ht]
\centering
\includegraphics[width=0.48\textwidth]{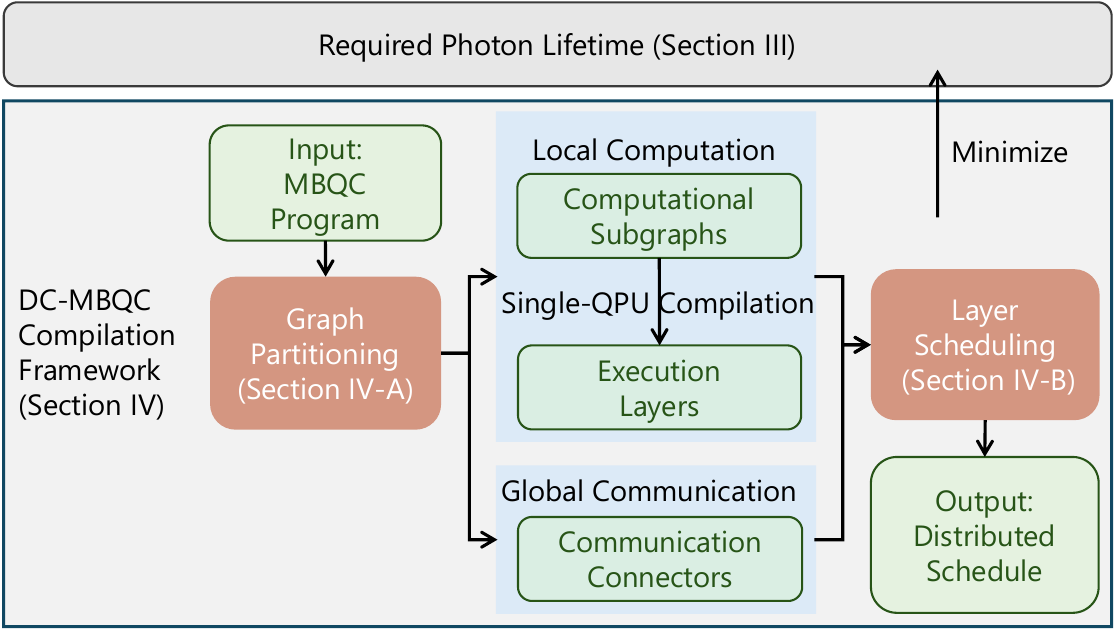}
\caption{An overview of our approach for minimizing the required photon lifetime in distributed MBQC.}
\label{fig:pipeline}
\end{figure}

To address the aforementioned challenges and pave the way for scalable MBQC, we propose a comprehensive solution that systematically formalizes and alleviates the burden on photon storage resources via distributed compilation (\fig{pipeline}). Our main contributions in this framework are listed as follows.
\begin{itemize}
    \item We formalize the \textbf{required photon lifetime} (\sec{lifetime}), which is the first metric to unify all sources of photon storage duration in delay lines, providing a fundamental target for optimizing MBQC compilers.
    
    \item We propose \textbf{\name} (\sec{framework}), the first compilation framework enabling distributed MBQC, with a primary focus on minimizing the required photon lifetime. Our framework decomposes the compilation into two key sub-problems: adaptive graph partitioning (\sec{graph-partitioning}) and layer scheduling (\sec{layer-mapping}), for which we develop highly effective algorithms. The framework's modular design ensures compatibility with various single-QPU compilers, paving the way for future integration with quantum error correction.
    
    \item Our experimental results strongly validate the effectiveness of our distributed approach. With 8 fully-connected QPUs, \name\ reduces the required photon lifetime by up to \textbf{7.46$\times$} and achieves a \textbf{6.82$\times$} speedup, demonstrating a viable path towards scalable and efficient photonic MBQC systems.
\end{itemize}

\section{Background} \label{sec:background}
\subsection{Photonic MBQC Program} \label{sec:mbqc}
An MBQC program is conceptually defined by two components: a \textit{graph state} and a corresponding \textit{measurement pattern}. The foundation of the computation is the graph state, a large-scale, multi-qubit entangled state constructed upon an undirected graph $G=(V, E)$. In this graph, each vertex $i \in V$ represents a qubit, and each edge $(i, j) \in E$ represents an entanglement between the corresponding pair of qubits. The state is uniquely defined as the $+1$ eigenstate of a set of commuting stabilizer operators:
\begin{equation*}
    K_i = X_i \prod_{j \in N(i)} Z_j, \quad \forall i \in V,
    \label{eq:graph-stabilizer}
\end{equation*}
where $X_i$ and $Z_j$ are the Pauli-X and Pauli-Z operators acting on qubits $i$ and $j$ respectively, and $N(i)$ is the set of neighbors of vertex $i$ in the graph $G$.

Computation in MBQC proceeds by executing the measurement pattern on the qubits of the graph state. A measurement pattern is a sequence of adaptive single-qubit measurements $M_i^\alpha$, where the parameter $\alpha$ of later measurements may depend on the outcomes of earlier ones. $M_i^\alpha$ represents a destructive measurement on qubit $i$, projecting it onto the basis $\{\ket{\pm_{\alpha}}_i = (\ket{0} \pm e^{i\alpha}\ket{1})/\sqrt{2}\}$ and yields a probabilistic outcome $s_i \in \{0, 1\}$.

This single, classical outcome $s_i$ is immediately fed forward to adapt subsequent measurements. This adaptation is Pauli \textit{byproduct corrections} ($X$ or $Z$) on subsequent neighbor qubits. The correction type is pre-determined by the MBQC program's design while the real-time outcome $s_i$ acts as a classical trigger, determining if that pre-defined correction is applied (e.g., when $s_i=1$) to ensure the computation remains deterministic. These corrections proceed in a single run and do not require multiple shots, since the correction logic relies on this single bit $s_i$, not on a probability distribution. 
Practically, these corrections are not applied as new gates; rather, they are absorbed by the measurement parameter $\alpha$.
A qubit to be measured by $M_i^{\alpha}$ with corrections $X^s$ and $Z^t$ ($s, t \in \{0, 1\}$) equals a qubit measured by
$$
M_i^{\alpha}X_i^sZ_i^t = M_i^{(-1)^s\alpha+t\pi}.
$$
These dependencies are formally captured by a directed acyclic graph, the \textit{dependency graph} $G'=(V, E')$. In this graph, an edge $(i, j)$ signifies that the measurement basis of qubit $j$ depends on the outcome of qubit $i$. These dependencies are classified as X-dependencies or Z-dependencies, corresponding to the type of Pauli correction required. Z-dependencies can be propagated to the end of the computation and handled classically via a technique known as signal shifting~\cite{BROADBENT20092489}, thereby removing them from the real-time dependency constraints.

This model describes MBQC at an abstract, logical level. From the practical perspective, most of MBQC programs require a huge graph state which is almost impossible to be entirely created at the beginning of computation. This fundamental difficulty motivates an alternative paradigm: generating the graph state dynamically and consuming its qubits incrementally~\cite{Browne_2005, Nielsen_2004}, which forms the basis of the photonic implementation discussed in \sec{mbqc-assumptions} and \sec{compiler}.

\subsection{Photonic MBQC Architecture} \label{sec:mbqc-assumptions}
As it is extremely hard to create and preserve a large graph state beforehand, photonic architectures overcome this challenge by dynamically generating the graph state during the computation. The core strategy is to start from small, standardized, and easily prepared \textit{resource states}, such as the ring or star shape states shown in \fig{rs}~(a). 
These states are then progressively combined together into a larger entangled structure by \textit{fusion}. As depicted in \fig{rs}~(b), a fusion operation consumes one photon from each of two resource states, entangling the neighbors of the original photons. Any remaining, unnecessary photons in the constructed graph can be disentangled from the computation by a simple Z-basis measurement, which, due to signal shifting~\cite{BROADBENT20092489}, does not introduce real-time dependencies. By this method, any arbitrary graph state can be created incrementally.

This dynamic generation process is orchestrated by the hardware architecture shown in \fig{mbqc-arch}. The primary components include a 2D-arranged grid of resource state generators (RSGs), a collection of fusion devices, and measurement devices. In each system clock cycle, every RSG produces a resource state. The set of all resource states generated simultaneously in one cycle constitutes a \textit{resource state layer}, establishing a discrete time dimension for the computation.

\begin{figure}[ht]
\centering
\includegraphics[width=0.45\textwidth]{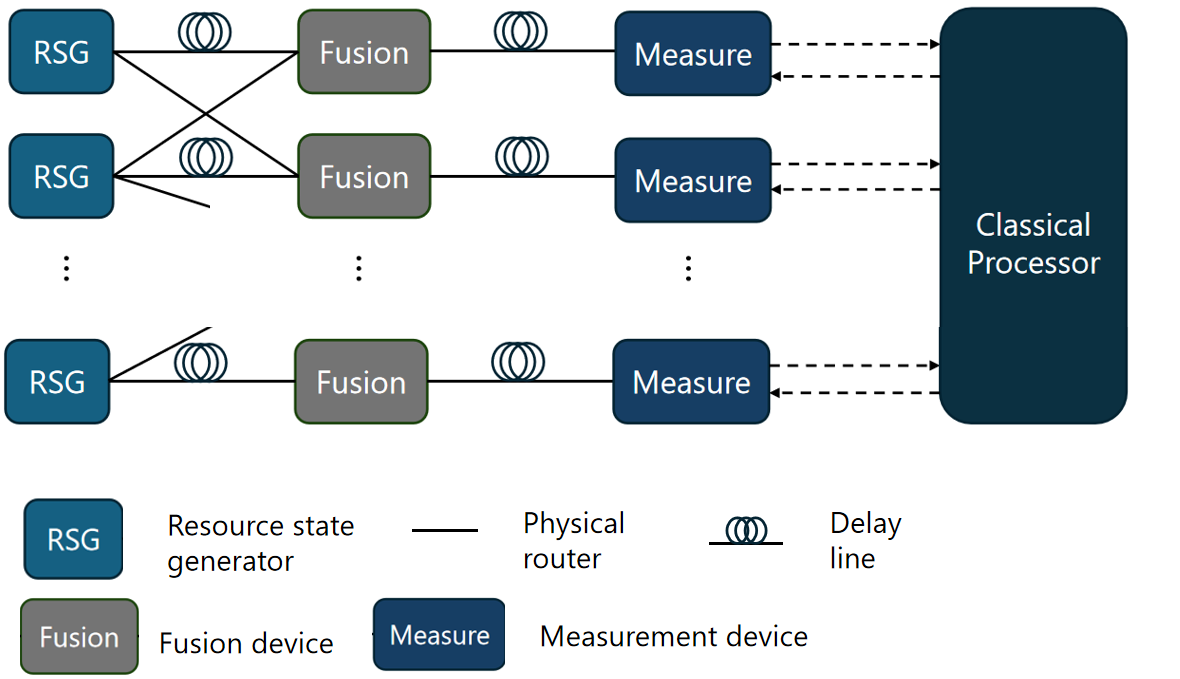}
\caption{MBQC architecture with $1$-dimensional arranged Resource State Generators (RSGs).}
\label{fig:mbqc-arch}
\end{figure}

Because of hardware constraints, many architectures assume that only local fusions are allowed. Typically, \textit{intra-layer fusion} is permitted between neighboring RSGs within the same layer, and \textit{inter-layer fusion} is permitted between subsequent layers from the same RSG~\cite{zhang2023oneq, zhang2024oneperc}. To realize arbitrary, long-range logical connections that are not directly supported by these local fusions, a fundamental hardware capability known as \textit{routing} is employed. As illustrated in \fig{rs}~(c), routing uses a chain of fusions to effectively create entanglement between two distant qubits.

These architectural constraints and measurement dependencies outlined in \sec{mbqc} pose a key system-level challenge: photon storage. Photons often need to wait for synchronization in fusion or for classical signals determining measurement bases, physically realized with delay lines (optical fibers calibrated to specific clock cycles). 
However, prolonged storage directly increases the risk of photon loss, a failure mode that fatally compromises program fidelity, especially at realistic clock rates. Despite this critical impact, the relationship between compilation choices and storage-induced fidelity degradation has lacked a formal, high-level metric. We address this gap by proposing the \textit{required photon lifetime} in \sec{lifetime} as a fundamental metric to evaluate and optimize compiler performance.

\begin{figure}[ht]
\centering
\includegraphics[width=0.45\textwidth]{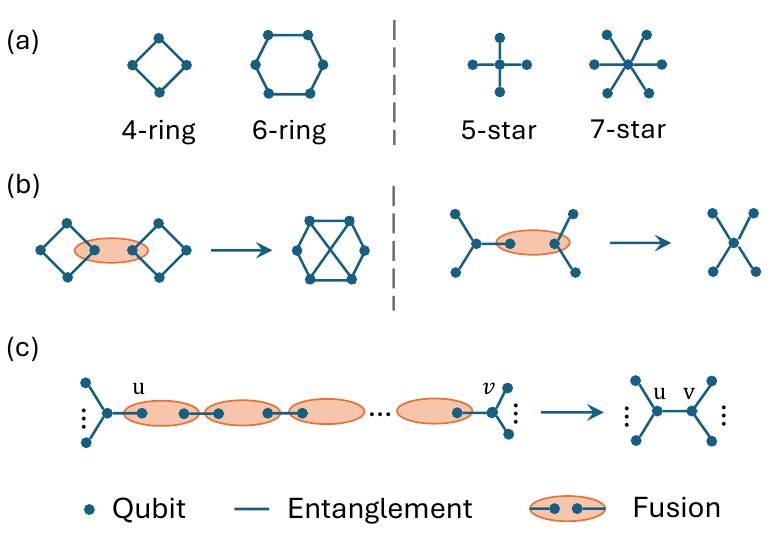}
\caption{(a) Different resource state graphs;
(b) Consequence of fusion on resource states;
(c) Illustration of routing. After the fusion process, the photon $u$ and $v$ are connected, forming an entangled pair.}
\label{fig:rs}
\end{figure}

To enable distributed computation, inter-QPU connections can be established via heralded entanglement~\cite{Barrett_2005, Hucul_2014}. By this method, the generation of communication resource is physically and logically decoupled from the main computational flow. Through this mechanism, ancillary photons are generated by dedicated communication ports that are often located at the edges of a QPU chip. An ancillary photon is promoted to a usable communication resource upon the arrival of a classical heralding signal confirming a successful remote connection. This ancillary photon can then be fused with a computational photon to bridge the two QPUs.

\subsection{Photonic MBQC Compiler}\label{sec:compiler}

The core task of a photonic MBQC compiler is to map the abstract logical  MBQC program, defined by a graph state and measurement pattern, onto the constrained physical hardware. This process can be understood in two main stages:  problem abstraction and spatio-temporal mapping. 

In the first stage, the compiler establishes a model of the computation, with different compilers using different abstraction strategies. One approach is to abstract the requirement of the input logical program into a form that is suitable for direct hardware mapping. For instance, the OneQ compiler~\cite{zhang2023oneq} constructs a fusion graph, where each node corresponds to a physical resource state and each edge represents a fusion operation. This representation directly describes the physical actions the hardware needs to perform. An alternative approach is to abstract the physical hardware per se, including OnePerc~\cite{zhang2024oneperc} and OneAdapt~\cite{zhang2025oneadaptadaptivecompilationresourceconstrained} that create a logical intermediate representation (IR).
This IR can be viewed as a discrete 3D (2D-spatial + 1D-temporal) grid, consisting of a time-ordered sequence of 2D logical layers. In this grid, nodes represent qubits and edges represent allowed entanglement between corresponding qubits.
This allows the input graph state to be mapped directly onto this logical IR, while the execution task on probabilistic hardware is handled by a separate mechanism named 2D-renormalization~\cite{zhang2024oneperc}. The viability of this logical abstraction stems from a key finding: the number of physical layers needed to produce one reliable logical layer (the PL Ratio) stabilizes around a constant value~\cite{zhang2024oneperc}. This allows high-level compilers to plan and optimize at the logical layer, decoupling the mapping problem from the hardware's inherent randomness.
Thus, this first stage transforms the compilation problem into a well-defined problem: mapping an input \textit{computation graph} onto a constrained \textit{3D resource grid}.

In the second stage, the compiler solves this mapping problem. It assigns each node of the computation graph onto a node in the 3D spatio-temporal grid, ensuring every pair of connected nodes in the computation graph still connected in the 3D grid, directly or through other nodes. The final output is a sequence of \textit{execution layers}, which precisely dictates the physical or logical operations for each clock cycle. Executing this sequence of layers completes the MBQC program.

\section{Required Photon Lifetime}\label{sec:lifetime}
In this section, we propose the concept of required photon lifetime $\tau_{\photon}$ as a key metric to measure the performance of a compiler on photonic MBQC architectures, moving beyond using only program execution time and capturing the critical hardware constraint of finite photon storage time.

Photons in a photonic MBQC architecture can be divided into the following three different kinds according to their roles in an MBQC program, each giving rise to required lifetime subject to distinct temporal constraints as illustrated in \fig{lifetime}.
\begin{itemize}
    \item \textit{Fusee: }A photon sent to a fusion device and fused with another fusee. A pair of fusees have to be synchronized at the fusion device. Consequently, the fusee on the earlier execution layer have to be stored in the delay line to await the arrival of its fusion partner.
    \item \textit{Measuree: }A photon sent to a measurement device and measured with a determined basis. The basis may depend on the measurement outcomes of other measurees. Therefore, it must be stored in a delay line until all dependencies are resolved and its basis is determined.
    \item \textit{Removee: }A photon that is unnecessary in the program and removed via a measurement in the Pauli-Z basis. Since Z-dependencies can be propagated to the end of the computation and handled classically by signal shifting~\cite{BROADBENT20092489}, removees do not contribute to the required photon lifetime.
\end{itemize}

We define \textit{required photon lifetime} of a compiled MBQC program as the maximum number of clock cycles that a photon has to be stored in the delay line. \alg{rl-calculation} shows the detailed calculation. In this algorithm, $G$ is the dependency graph of the MBQC program and execution layers are compiler output. $\NetIndex(u)$ is the index of the corresponding execution layer of photon $u$. In calculation of $\tau_{\measuree}$, $\Call{TopoSort}{G}$ sorts all nodes in $G$ in topological order. By this method, $\maxparent[u]$ computes the earliest measurable time of photon $u$, and $\tau_u$ becomes the required lifetime. We assume that a photon arrives the measurement device one clock cycle after it is generated, and computation of measurement basis requests one clock cycle.

\begin{figure}[ht]
\centering
\includegraphics[width=0.45\textwidth]{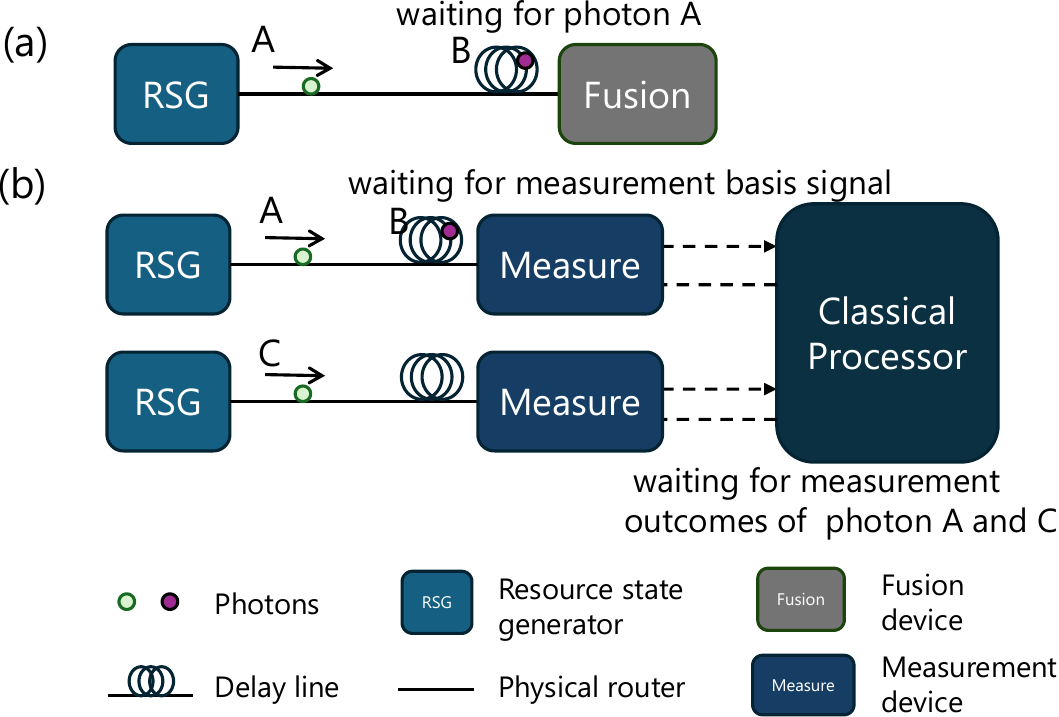}
\caption{
Illustration of primary sources of required photon lifetime. (a) A fusee (photon B) must wait for its fusion partner (photon A) when they are generated in different execution layers. 
(b) A measuree (photon B) must be stored while waiting for classical information. This wait occurs because its measurement basis depends on measurement outcomes of other photons (A and C).
}\label{fig:lifetime}
\end{figure}

\begin{algorithm}[tb]
\caption{Required Photon Lifetime Calculation}
\label{alg:rl-calculation}
\begin{algorithmic}[1]
\Require Execution layers; dependency graph $G$
\Ensure Required photon lifetime $\tau_{\photon}$
\Function{ComputeLifetime}{$G, \text{execution layers}$}
  \Statex  \Comment{Part 1: compute fusee required lifetime $\tau_{\fusee}$}
    \State $\tau_{\fusee} \gets 0$
    \For {each pair of fusees $(u, v)$}
        \State $\tau_{(u,v)} \gets \abs(\NetIndex(u)-\NetIndex(v))$
        \State $\tau_{\fusee} \gets \max(\tau_{\fusee}, \tau_{(u,v)})$
    \EndFor
   \Statex \Comment{Part 2: compute measuree required lifetime $\tau_{\measuree}$}
    \State $\tau_{\measuree} \gets 0$
    \For {each node $u$ in $\List(\Call{TopoSort}{G})$}
        \State $\maxparent[u] \gets \NetIndex(u) + 1$
        \For {each parent $v$ in $\Parent(u)$}
            \State $\maxparent[u] \gets \max(\maxparent[v] + 1,\maxparent[u])$
        \EndFor
        \State $\tau_u \gets \maxparent[u]- \NetIndex(u)$

        \State $\tau_{\measuree} \gets \max(\tau_{\measuree}, \tau_u)$
    \EndFor
   \Statex \Comment{Part 3: return required photon lifetime}
    \State $\tau_{\photon} \gets \max(\tau_{\fusee}, \tau_{\measuree})$
    \Return $\tau_{\photon}$
\EndFunction
\end{algorithmic}
\end{algorithm}

Distributed compilation introduces another kind of photon, \textit{connectors}, which is defined in \sec{framework} and the corresponding required lifetime calculation is addressed in \df{LSP}.

\section{\name\ Framework}\label{sec:framework}
To overcome the scalability limitations of single QPU, various distributed quantum computing approaches emerge, %
including circuit distribution and circuit cutting~\cite{Barral_2025}.

Circuit distribution seeks to execute a large quantum circuit that exceeds the capacity of a single QPU by mapping it onto multiple QPUs connected via a quantum network. The computation is performed through unitary evolution on long-lived, stateful qubits, while inter-QPU communication is realized via remote operations such as teledata and telegate. The core compiler challenge lies in allocating logical qubits across QPUs and implementing remote two-qubit gates under resource constraints~\cite{caleffi2024distributed, Ferrari_2021}. In contrast, MBQC programs are performed via single-qubit measurements on a large, entangled graph state, and interconnection between QPUs can be realized through photonic fusion. This fundamental difference in the underlying computational model necessitates a dedicated compiler framework, with existing techniques for circuit distribution not directly applicable.

The other main approach, circuit cutting, aims to decompose a large quantum circuit into smaller subcircuits that are executable on limited devices, then reconstructs global observable results through classical post-processing~\cite{Peng_2020, Tang_2021}. However, this reconstruction typically introduces exponential sampling overhead~\cite{jing2025circuit}. MBQC can avoid this bottleneck by using fusion to establish inter-QPU communication, making the core focus different from circuit cutting.

As far as we know, this paper is the first one to propose a framework of physically distributed computation tailored to MBQC. A distributed compilation framework has to address the following two challenges:
\begin{itemize}
    \item \textbf{Workload distribution.} Since input MBQC programs are modeled as computation graphs by compilers, the workload distribution is to partition the computation graph across multiple QPUs, balancing workload and minimizing costly inter-QPU communication. 
    \item \textbf{Communication and synchronization.} The partitioning of the computation graph inherently severs logical entanglements, creating pairs of \textit{connectors}, which are qubits at the boundaries of the newly formed subgraphs. This gives rise to a key challenge that efficiently re-establishing these connections across QPUs.
\end{itemize}

While this can be physically achieved by fusing connectors with ancillary communication resources (see \sec{mbqc-assumptions}), routing connectors to communication resources remains a formidable challenge. A naive approach might dedicate fixed physical ``edge zones'' on each QPU for this purpose as illustrated in \fig{zone}~(a). However, this imposes rigid spatial constraints that would severely limit the optimization capacity of placement-agnostic compilers such as OneQ~\cite{zhang2023oneq}, OnePerc~\cite{zhang2024oneperc}, and OneAdapt~\cite{zhang2025oneadaptadaptivecompilationresourceconstrained}. To overcome this, we leverage the temporal dimension by introducing \textit{connection layers} in the 3D resource grid. Connectors are routed to these connection layers by inter-layer fusion and then routed to corresponding communication resource by intra-layer fusions of computation resource on these special layers, as illustrated in \fig{zone}~(b). This strategy achieves a crucial spatio-temporal decoupling, liberating the single-QPU compiler to optimize the computational layout without being constrained by the fixed physical locations of communication ports.

\begin{figure}[ht]
\centering
\includegraphics[width=0.45\textwidth]{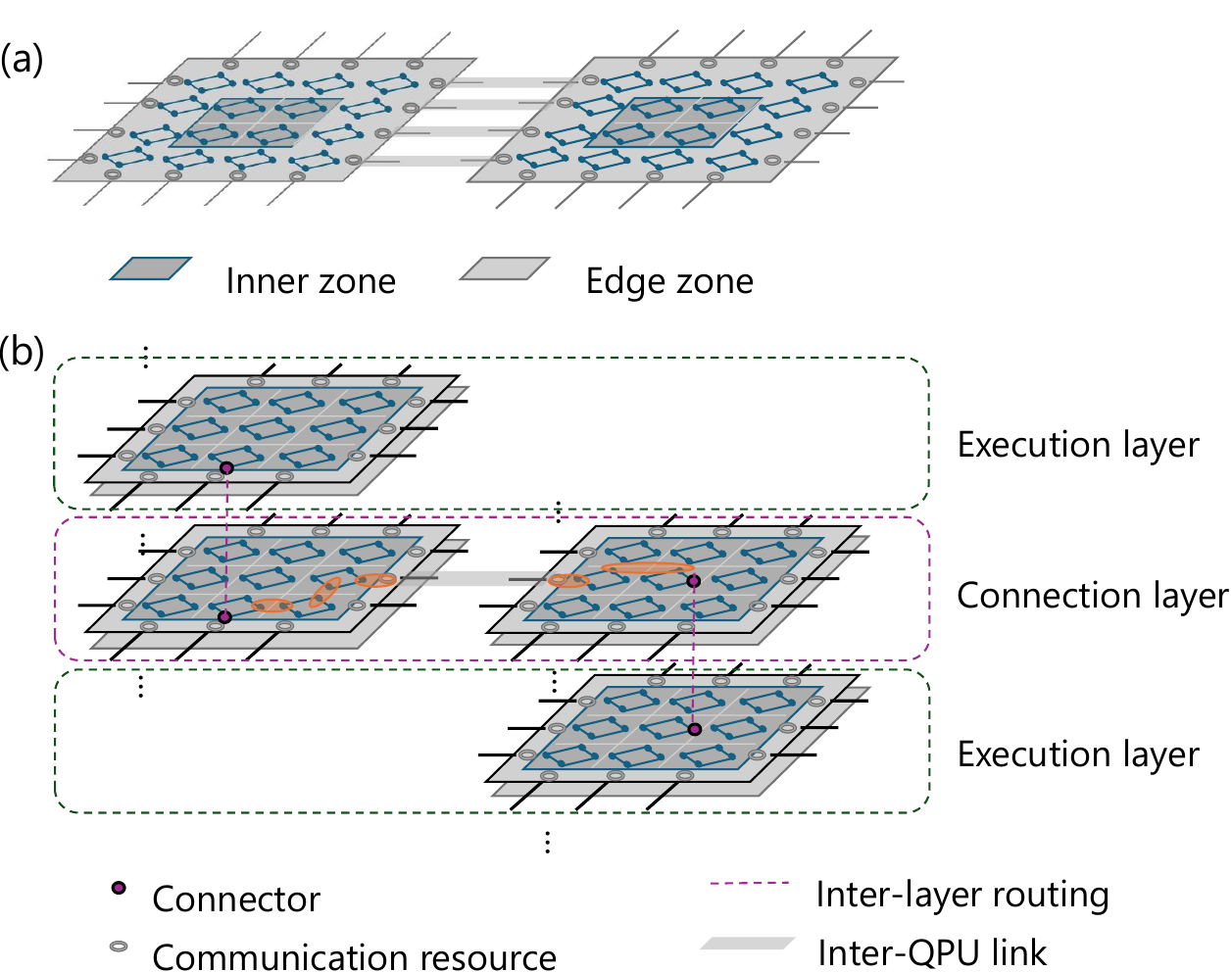}
\caption{(a) Illustration of fixed edge zones. (b) Illustration of inter-QPU communication via connection layer.}\label{fig:zone}
\end{figure}
 
However, on a fixed connection layer, the number of connectors that can be routed to communication resources concurrently is essentially an instance of the edge-disjoint paths (EDP) problem on a grid. This problem, which involves finding the maximum number of non-overlapping paths between specified terminals, is known to be NP-hard~\cite{Garey1979JohnsonCA}. In addition, the number and location of available communication resources are dependent to specific hardware. Since our framework focuses on high-level design and scheduling, we abstract the complexity of concurrent on-layer routing into a single, practical hardware parameter, the connection capacity $K_{\max}$, which represents the maximum number of concurrent connections a single layer can support.

Although NP-hardness exists, we can provide upper and lower bounds for this $K_{\max}$. By Menger's theorem, the maximum number of disjoint paths, which is an upper bound of $K_{\max}$, is limited by the minimum cut of the connector set and the communication resource set on the grid. This implies $K_{\max} = O(L)$ for a layer of size $L \times L$. A straightforward lower bound is $K_{\max} \geq 4$ since there are four edges. Experimentally, we conduct sensitivity analysis in \sec{sensitivity}, demonstrating high-performance of our \name\ compiler when $K_{\max} = 4$.

Based on aforementioned abstractions, our \name\ framework decomposes the distributed compilation into two sub-problems:
\begin{itemize}
    \item \textbf{Graph partitioning.} In \sec{graph-partitioning}, we propose an adaptive graph partitioning algorithm to partition the computation graph, navigating the trade-off of load balancing and communication minimizing, as well as preserving local structure of subgraphs for efficient single-QPU compilation.
    \item \textbf{Layer scheduling.} In \sec{layer-mapping}, we formalize the challenge about inter-layer scheduling as the layer scheduling problem, prove the NP-hardness and provide efficient heuristic solution.
\end{itemize}

\subsection{Graph Partitioning}\label{sec:graph-partitioning}
The first core task of our distributed compiler is to partition the overall computation across multiple QPUs. While it is conceptually similar to graph partitioning in the dominant circuit model~\cite{automated2019, timeslice2020}, our MBQC context is fundamentally different. Circuit-based partitioners operate on an interaction graph of persistent qubits, with the goal of minimizing the number of non-local operations such as teleportation. 
In contrast, our MBQC model is built upon a graph state of transient photons, with the primary optimization objective being the required photon lifetime at the final stage. This shifts the focus beyond mere inter-QPU communication demands, necessitating consideration of subsequent single-QPU compilation and layer scheduling.
Therefore, existing gate-based approaches are not directly applicable in distributed MBQC.

In our framework, we partition the 
computation graph with each node representing a fundamental resource unit; thus, achieving a balanced partition directly contributes to workload balancing.
Within the context of MBQC, an effective partitioning algorithm must additionally co-optimize two often competing objectives:
\begin{itemize}
    \item Minimized Communication: The number of edges cut between partitions should be minimized, as each cut corresponds to a costly instance of inter-QPU communication.
    \item Preserved Local Structure: The partitions should maintain high quality in their internal structure. Subgraphs with strong internal connectivity allow for more efficient compilation by single-QPU compilers~\cite{zhang2023oneq}.
\end{itemize}

These two objectives present a trade-off. Standard $k$-way partitioning algorithms~\cite{KARYPIS199896} excel at minimizing cuts while enforcing a strict load balance, but they ignore the subgraph structure. Conversely, community detection algorithms such as~\cite{Blondel_2008, Traag_2019} are designed to maximize modularity~\cite{PhysRevE.69.066133}, which is a widely recognized metric that quantifies the quality of a partition by measuring the density of connections within subgraphs (communities) relative to connections between them. However, these algorithms do not guarantee a fixed number of partitions or workload balance.

To navigate this trade-off, we propose an adaptive graph partitioning algorithm (\alg{graph-partition}) that integrates the strengths of both approaches and performs a heuristic search through the imbalance-modularity trade-off space. The algorithm begins with a perfectly balanced partition ($\alpha=1$) generated by the METIS library, which implements a standard multilevel k-way partitioning scheme~\cite{KARYPIS199896}. It then iteratively probes for better partitions by relaxing the balance constraint ($\alpha$), using a multiplicative step factor ($\gamma$). The search accepts a new, less balanced partition only if it yields a modularity gain greater than $\epsilon_Q$, and terminates when the modularity gain stagnates or the maximum imbalance ($\alpha_{\max}$) is reached. This process efficiently finds a partition that is both well-structured (high $Q$) and reasonably balanced (less than $\alpha_{\max}$).

\begin{algorithm}
\caption{Adaptive Graph Partitioning}\label{alg:graph-partition}
\begin{algorithmic}[1]
\State \textbf{Input:} computation graph $G$,  modularity improvement threshold $\epsilon_Q$, maximum imbalance factor $\alpha_{\max}$, learning rate $\gamma$
\State \textbf{Output:} best partition $P_{\best}$

\State $\alpha \gets 1$
\State $Q_{\best} \gets -1$
\While {$\text{True}$}
    \State $P \gets \text{Partition}(G, \alpha)$
    \State $Q \gets \text{Modularity}(P)$
    \If{$Q > Q_{\best}$}
        \State $Q_{\best} \gets Q$, $P_{\best} \gets P$
    \EndIf
    \State $\Delta Q \gets Q - \text{previous modularity}$    
    \If{$\Delta Q > \epsilon_Q$ and $\alpha < \alpha_{\max}$}
        \State $\alpha \gets \alpha \cdot \gamma$
    \ElsIf{$\Delta Q < -\epsilon_Q$}
        \State $\alpha \gets \alpha/ \gamma$
    \Else
        \State \textbf{break}
    \EndIf
\EndWhile
\end{algorithmic}
\end{algorithm}

\subsection{Layer Scheduling}\label{sec:layer-mapping}
We first formalize the layer scheduling problem as follows.
In this formalism, we model the compilation output as two distinct types of tasks: Main Tasks ($J_{i,j}$), which represent the execution layers compiled for a specific QPU, and Synchronization Tasks ($S_k$), which represent the inter-QPU communication events.
\begin{definition}[Layer Scheduling Problem]\label{def:LSP}
The problem is to optimize the scheduling for two classes of tasks, main task and synchronization task, on a set of $n$ QPUs, $\mathcal{Q} = \{Q_1, \ldots, Q_n\}$, over a discrete time horizon $\mathcal{T} = \{1, \ldots, T\}$. An identical connection parameter $K_{\max}$ is given.

\paragraph{Tasks}
\begin{itemize}
    \item Main task. Each QPU $Q_i$ is statically assigned a set of main tasks, $\mathcal{J}_i = \{J_{i,1}, \ldots, J_{i,m_i}\}$.
    \item Synchronization task. A separate set of synchronization tasks $\mathcal{S} = \{S_1, \ldots, S_K\}$ is defined, where each $S_k$ is associated with a pair of main tasks $(J_{i,j}, J_{i',j'})$ on distinct processors ($i \neq i'$).
\end{itemize}

\paragraph{Decision Variables}
The decision variables are the start times for each task:
\begin{itemize}
    \item $j_{ij} \in \mathcal{T}$: The start time of main task $J_{i,j}$.
    \item $s_k \in \mathcal{T}$: The start time of synchronization task $S_k$.
\end{itemize}

\paragraph{Constraints}
The schedule must adhere to the following constraints:
\begin{itemize}
    \item \textbf{Machine Exclusivity:} A processor $Q_i$ can execute at most one main task or up to $K_{\max}$ synchronization tasks at any given time.
    \begin{align*}
      &  \forall i \in \{1,\ldots,n\},\ \forall t \in \mathcal{T},\\
& \sum_{j=1}^{m_i} \mathbb{I}(j_{ij} = t)  + \left\lceil \frac{1}{K_{\max}} \sum_{S_k \text{ involves } Q_i} \mathbb{I}(s_k = t) \right\rceil \le 1
    \end{align*}
    where $\mathbb{I}(\cdot)$ is the indicator function.
    \item \textbf{Synchronization:} A synchronization task $S_k$ associated with processors $M_i$ and $M_{i'}$ must execute simultaneously. This is inherently enforced by assigning a single decision variable $s_k$ for its start time.
\end{itemize}

\paragraph{Objective Function}
The objective is to find a schedule that minimizes the required photon lifetime defined in \sec{lifetime}. 
As mentioned above, main tasks are execution layers in local computation and synchronization tasks deal with remote communication.
The required lifetime of local computation $\tau_{\local}$ is calculated by \alg{rl-calculation}, while layer index is replaced by the start time $j_{ij}$ of the corresponding main task. The required lifetime of remote communication $\tau_{\remote}$ is essentially caused by inter-layer fusion between the corresponding execution layer and connection layer. Consequently,
$$
\tau_{\remote} = \max_{S_k \text{ associated with } J_{i,j}}\abs(s_k - j_{i,j}).
$$
The final objective is $\min \max (\tau_{\local}, \tau_{\remote})$.
\end{definition}

\begin{theorem}[NP-hardness of Layer Scheduling]\label{thm:NP}
Layer scheduling is NP-hard and, furthermore, cannot be approximated within any constant factor in polynomial time unless P = NP.
\end{theorem}
\begin{proof}
We prove that, a single-QPU, decision version of layer scheduling problem (LSP) is already NP-hard by a polynomial-time reduction from the decision version of graph bandwidth problem (GBP). GBP asks: given a graph $G=(V, E)$ and an integer $k$, does there exist a linear layout $f$ such that the bandwidth is at most $k$?

Given an arbitrary instance of GBP $(G, k)$, we construct an instance of single-QPU LSP in polynomial time. We define a set of main task $\mathcal{J} = \{J_u, \forall u \in V\}$, and introduce a pair of fusees on the corresponding execution layers for each edge $(u, v) \in E$. This construction is polynomial in the size of the input graph.

We now show that the graph $G$ has a bandwidth of at most $k$ if there exists a valid schedule for the constructed LSP instance with a maximum required lifetime of at most $k$. 
Assuming a solution $j: \mathcal{J} \rightarrow \mathcal{T}$ exists, we can define a layout for GBP as $f(u) = j(J_u)$ for all $u \in V$. For any edge $(u,v) \in E$, $|f(u) - f(v)| \leq k$ since the required lifetime caused by fusee pairs on $J_u$ and $J_v$ is at most $k$. This layout demonstrates that the bandwidth of $G$ is at most $k$.

Therefore, $\GBP \leq_p \LSP$, LSP is NP-hard due to the NP-hardness of GBP~\cite{Papadimitiou_1976} Furthermore, it is NP-hard to find an approximation within any constant ratio~\cite{743431}.
\end{proof}

This NP-hardness motivates our development of high-performance heuristics for efficient, high-quality scheduling. We first present a standard priority-based list scheduling heuristic as a baseline, followed by our bottleneck-driven iterative refinement (BDIR) algorithm, which leverages simulated annealing (SA) to systematically improve initial schedules by identifying and resolving performance bottlenecks. Standard priority-based list scheduling~\cite{Pinedo} is a natural starting point due to the sequential locality of main tasks, corresponding to execution layers from a single-QPU compiler passes. However, it often introduces costly bottlenecks by neglecting inter-QPU synchronization and lacking a global view of required photon lifetime. A naive strategy, e.g., relocating random tasks, fails to address performance bottlenecks and frequently violates dependency constraints. Designing effective neighborhood generation is challenging due to:
\begin{itemize}
    \item How to select a bottleneck task and its target destination? Due to the global min-max objective function, there exist conflicting optimization pressures from different cost sources, restricting the move of bottleneck tasks.
    \item How to execute this move without violating restrictive synchronization constraints or catastrophically disrupting the existing schedule structure? Synchronization constraints invalidate standard swap operator since synchronization tasks must be scheduled at the same time on two QPUs, thereby expanding the influence scope of swap operator.
\end{itemize}

To address these two challenges, we designed the Bottleneck-Driven Iterative Refinement (BDIR) algorithm (\alg{bdir}), a ``smart" neighborhood generation function embedded within a lightweight Simulated Annealing (SA) framework. The $\Call{GenerateNeighbor}{}$ function is not random; it precisely targets the schedule's primary bottleneck. First, $\Call{FindBottleneckTask}{}$ identifies the task $N$ responsible for the current maximum required photon lifetime ($\tau_{photon}$). Second, instead of moving $N$ randomly, the $\Call{CalculateBalancePoint}{}$ heuristic determines a near-optimal target time slot. It does this by evaluating the local cost ($N$'s contribution to $\tau_{photon}$) at different time slots, assuming all other tasks remain fixed at their current positions ($L_{current}$). The chosen slot $t$ is the ``temporal equilibrium point" that best balances the cost pressures from $N$'s preceding dependencies and its remote fusion partners. Third, the $\Call{PinAndReschedule}{}$ function executes the move. It pins $N$ to the target time $t$ and re-schedules all other tasks using the standard priority-based list scheduler. Crucially, to preserve the existing schedule structure, the priority for each re-scheduled task is set to its start time in the original schedule ($L_{current}$). This ensures the schedule's relative ordering is maintained while resolving any constraints violated by moving $N$. Finally, this ``smart" $\Call{GenerateNeighbor}{}$ function is wrapped in an SA framework. 
The initial schedule $L_{init}$ is generated using list scheduling, where a main task $J_{i,j}$'s priority is its index $j$, and a synchronization task $S_k$ (for $J_{i,j}, J_{i',j'}$) has priority $\frac{1}{2}(j + j')$.

\begin{algorithm}
\caption{Bottleneck-Driven Iterative Refinement (BDIR)}
\label{alg:bdir}
\begin{algorithmic}[1]
\State \textbf{Input: } SA initial temperature $T_0$, cooling rate $\alpha$, max iterations $I_{\max}$, initial schedule $L_{\init}$
\State \textbf{Output: } optimized schedule $L_{\best}$
\Statex
\Function{BDIR}{$T_0, \alpha, I_{\max}, L_{\init}$}
    \State $L_{\current} \gets L_{\init}$, $T \gets T_0$
    \State $L_{\best} \gets L_{\init}$, $c_{\best} \gets \Call{ComputeLifetime}{L_{\best}}$
    \For{$i = 1 \to I_{max}$}
        \State $L_{\new} \gets \Call{GenerateNeighbor}{L_{\current}}$ 
        \State $c_{\current} \gets \Call{ComputeLifetime}{L_{\current}}$
        \State $c_{\new} \gets \Call{ComputeLifetime}{L_{\new}}$
        \State $\Delta E \gets c_{\new} - c_{\current}$
        \If{$\Delta E \leq 0$ \textbf{or} $\Call{Random}{0,1} < e^{-\Delta E / T}$}
            \State $L_{\current} \gets L_{\new}$
        \EndIf
        \If{$c_{\current} < c_{\best}$}
            \State $L_{\best} \gets L_{\current}$, $c_{\best} \gets c_{\current}$
        \EndIf
        \State $T \gets \alpha \cdot T$
    \EndFor
    \State \textbf{return} $L_{\best}$
\EndFunction

\Statex
\Function{GenerateNeighbor}{$L_{\current}$}
    \State $N \gets \Call{FindBottleneckTask}{L_{\current}}$ 
    \State $t \gets \Call{CalculateBalancePoint}{L_{current}, N}$
    \State $L_{\new} \gets \Call{PinAndReschedule}{L_{\current}, N, t}$ 
    \State \textbf{return} $L_{\new}$
\EndFunction
\end{algorithmic}
\end{algorithm}

\section{Evaluation} \label{sec:evaluation}

\subsection{Experiment Setup} \label{sec:setup}

\noindent \textbf{Baseline.} 
The baseline method we compare to is the OneQ compiler~\cite{zhang2023oneq}, a recent compilation framework specifically tailored for photonic MBQC architecture.
While more recent frameworks such as OnePerc~\cite{zhang2024oneperc} and OneAdapt~\cite{zhang2025oneadaptadaptivecompilationresourceconstrained} have introduced further optimizations, OneQ established the foundational methodology for mapping MBQC programs to the 3D resource grid. Crucially, these successors build upon the core compilation stages defined in OneQ. Comparisons with OnePerc and OneAdapt are discussed in \sec{sota}.

\noindent \textbf{Metric.}
As discussed in \sec{intro}, prolonged photon storage in delay lines dramatically increases the probability of photon loss. This loss is a primary source of fatal errors that fatally compromise program fidelity, especially at realistic clock rates (see \fig{required-lifetime}). Therefore, our core metric is the required photon lifetime (\sec{lifetime}) , which serves as a critical, high-level proxy for this storage-induced fidelity degradation. We also consider \textit{execution time} as an important metric, which is the number of clock cycles needed to run an MBQC program.

\begin{table}[b]
\centering
  \caption{Benchmark programs.}
  \label{tab:benchmark}
  \begin{tabular}{c|c|c|c|c}
    \toprule
    Program & \#Qubits & Grid size & \#2Q gates & \#Fusion \\
    \midrule
    \midrule
    \multirow{4}{*}{VQE}  & 16    & $7\times7$  & 120 &  408  \\ 
                          & 36  & $11\times11$  & 630 & 2178  \\
                          & 81  & $17\times17$  & 3240 & 11280  \\
                          & 144 & $23\times23$ & 10296 & 35928\\
    \midrule
    \multirow{4}{*}{QAOA} & 16  & $7\times7$  & 47 & 487 \\ 
                          & 64  & $15\times15$  & 799 & 7316  \\
                          & 121  & $21\times21$  & 2843 & 25826  \\
                          & 196 & $27\times27$ & 7528 & 68141\\
    \midrule
    \multirow{4}{*}{QFT}  & 16    & $7\times7$  & 120 &  408  \\ 
                          & 36  & $11\times11$  & 630 & 2178  \\
                          & 81  & $17\times17$  & 3240 & 11280  \\
                          & 100 & $19\times19$ & 4950 & 64450\\
    \midrule
    \multirow{3}{*}{RCA} & 16  & $7\times7$  & 209 & 1108 \\ 
                          & 36  & $11\times11$  & 529 & 2808  \\
                          & 81 & $17\times 17$ & 1249 & 6633 \\
    \bottomrule
\end{tabular}
\end{table}

\noindent \textbf{Benchmark programs.} 
We select the Quantum Approximate Optimization Algorithm (QAOA)~\cite{farhi2014quantum}, Variational Quantum Eigensolver (VQE)~\cite{peruzzo2014variational}, Quantum Fourier Transform (QFT)~\cite{coppersmith2002approximate}, and Rippler-Carry Adder (RCA)~\cite{cuccaro2004new} as benchmarks, which consist of algorithmic primitives for quantum programming (QFT and RCA) and application orientated programs (QAOA for optimization problems and VQE for quantum chemistry simulation). 
In \tab{benchmark}, we list the number of qubits, the spatial grid size (i.e., the size of each 2D logical resource layer), the number of 2-qubit gates and the number of fusions, which are edges in the computation graph in OneQ.
For QAOA, the circuits are generated from the Max-Cut problem on random graphs. The graphs are generated by randomly selecting half of all possible edges.
For VQE, we use the hardware efficient ansatz~\cite{kandala2017hardware} with fully entangled layers, where each qubit pair gets connected through a CNOT gate and hence this leads to a quadratically increasing 2-qubit gate counts with respect to the number of qubits.
Similar gate counts are observed for QFT programs.
The largest cases we evaluated are the 196-qubit QAOA program which contains 7528 2-qubit gates and 68141 fusions.

\noindent \textbf{Evaluation Platform.} All compilation runtime was measured on a workstation equipped with an Intel(R) Xeon(R) Gold 6459C with operation system Ubuntu 22.04.5 LTS. All reported compilation time represents the average of five runs to ensure measurement stability.

\noindent \textbf{Parameter Settings.} For the adaptive graph partitioning (\alg{graph-partition}), we set the improvement threshold $\epsilon_Q = 0.01$ and step factor $\gamma = 1.02$. For the BDIR scheduler (\alg{bdir}), we used an initial temperature $T_0 = 10$, a cooling rate $\alpha = 0.95$, and $I_{max} = 20$ iterations. 
The two key parameters, connection capacity $K_{max}$ and maximum imbalance factor $\alpha_{max}$, were set to default values of $K_{max}=4$ and $\alpha_{max}=1.5$ for all main experiments, unless otherwise specified. A detailed sensitivity analysis of these two parameters is presented in \sec{sensitivity}. 

    \subsection{Experiment Result} \label{sec:result}

\begin{table*}[ht]
    \centering
    \caption{The results of \name\ and its relative performance to the baseline using 4 QPUs and 5-star RSG.}
    \label{tab:dcmbqc_results_4qpu}
    \begin{tabular}{c|c|c|c|c|c|c}
        \toprule
        Program-\#Qubits & Baseline Exec. Time& Our Exec. Time & Improv. Factor & Baseline Lifetime & Our Lifetime & Improv. Factor \\
        \midrule
        \midrule
        VQE-16  & 46  & 21 & 2.19 & 38  & 14 & 2.71 \\
        VQE-36  & 155 & 46 & 3.37 & 143 & 36 & \textbf{3.97} \\
        VQE-81  & 462 & 128 & 3.61 & 441 & 112 & 3.94 \\
        VQE-144  & 1022 & 278 & 3.68 & 1000 & 258 & 3.88 \\
        \cmidrule{1-7}
        QAOA-16 & 45  & 16 & 2.81 & 36  & 10 & 3.60 \\
        QAOA-64 & 284  & 83 & 3.42 & 265  & 70 & 3.79 \\
        QAOA-121 & 731 & 198 & 3.69 & 717 & 183 & 3.92 \\
        QAOA-196 & 1401 & 394 & 3.56 & 1370 & 369 & 3.71 \\
        \cmidrule{1-7}
        QFT-16  & 126 & 35 & 3.60 & 107 & 28 & 3.82 \\
        QFT-36  & 364 & 101 & 3.60 & 333 & 81 & 4.11 \\
        QFT-81  & 1000	& 293 &	3.41 & 975	& 264 &	3.69\\
        QFT-100  & 1295 & 386 & 3.35 & 1272 & 344 & 3.70 \\
        \cmidrule{1-7}
        RCA-16  & 103 & 27 & \textbf{3.81} & 90 & 26 & 3.46 \\
        RCA-36  & 121 & 34 & 3.56 & 123 & 57 & 2.16 \\
        RCA-81 & 158 & 44 & 3.59 & 218 & 135 & 1.61 \\
        \bottomrule
    \end{tabular}
\end{table*}

\begin{table*}[ht]
    \centering
    \caption{The results of \name\ and its relative performance to the baseline using 8 QPUs and 4-star RSG.}
    \label{tab:dcmbqc_results_8qpu}
    \begin{tabular}{c|c|c|c|c|c|c}
        \toprule
        Program-\#Qubits & Baseline Exec. Time & Our Exec. Time & Improv. Factor & Baseline Lifetime & Our Lifetime & Improv. Factor \\
        \midrule
        \midrule
        VQE-16  & 54 & 14 & 3.86 & 45  & 10  & 4.50\\
        VQE-36  & 155 & 40 & 5.17 & 140 & 23 & 6.09 \\
        VQE-81  & 471 & 72 & 6.54 & 441 & 66 & 6.78 \\
        VQE-144  & 1016 & 149 & \textbf{6.82} & 991 & 137 & 7.23 \\
        \cmidrule{1-7}
        QAOA-16 & 48  & 14 & 3.43 & 39  & 11 & 3.55 \\
        QAOA-64 & 273  & 56 & 4.88 & 253  & 43 & 5.88 \\
        QAOA-121 & 712 & 116 & 6.14 & 696 & 106 & 6.57 \\
        QAOA-196 & 1401 & 204 & 6.87 & 1372 & 186 & 7.38 \\
        \cmidrule{1-7}
        QFT-16  & 120 & 21 & 5.71 & 110 & 18 & 6.11 \\
        QFT-36  & 356 & 54 & 6.59 & 329 & 46 & 7.15 \\
        QFT-81  & 1036 & 152 & 6.81 & 993 & 133 & \textbf{7.46} \\
        QFT-100  & 1250 & 186 & 6.72 & 1221 & 167 & 7.31 \\
        \cmidrule{1-7}
        RCA-16  & 103 & 22 & 4.68 & 90 & 29 & 3.10 \\
        RCA-36  & 121 & 18 & 6.72 & 123 & 54 & 2.28 \\
        RCA-81 & 138 & 27 & 5.11 & 203 & 126 & 1.61\\
        \bottomrule
    \end{tabular}
\end{table*}

We first show results for comparing the baseline OneQ compiler and our DC-MBQC compiler (using 4 QPUs and 5-star resource state) in terms of execution time and required photon lifetime in \tab{dcmbqc_results_4qpu}.
QneQ compiler is designed for a single QPU, while the purposed compiler design utilizes multiple QPUs.
It can be seen that DC-MBQC consistently outperforms the baseline in both execution time and required photon lifetime up to $3.97\times$ (for 36-qubit VQE program).

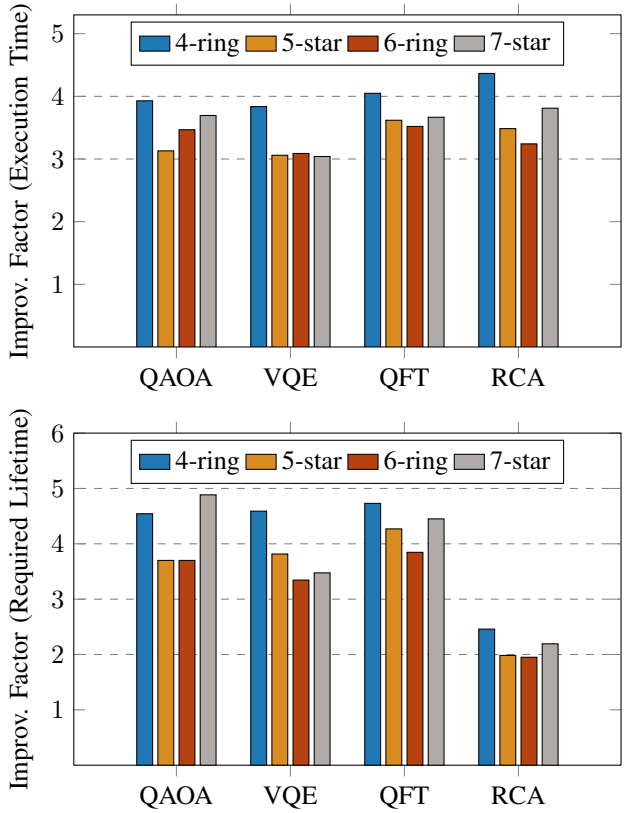
\begin{figure}[ht]  
\centering 
\begin{tikzpicture}
    \begin{axis}[
        width=\columnwidth, %
        height=6cm,       %
        ybar,
        bar width = 6pt,
        xtick={0, 20, 40, 60},
        ytick={1,2,3,4,5},
        xticklabels={QAOA, VQE, QFT, RCA},
        ymin=0, ymax=5.3,
        ylabel={Improv. Factor (Execution Time)},
        legend style={at={(0.5,0.98)}, anchor=north, legend columns=-1},
        enlarge x limits=0.3,
        legend image code/.code={
        \draw[fill=#1] (0cm,-0.1cm) rectangle (0.4cm,0.1cm);}
    ]   
        \draw [dashed, thin, color=gray](-20,3) -- (100,3);
        \draw [dashed, thin, color=gray](-20,4) -- (100,4);
        \addplot[fill=customblue] coordinates {(0, 3.92857) (20, 3.83673) (40, 4.04878) (60, 4.36585)};
        \addplot[fill=orange] coordinates {(0, 3.13) (20, 3.06) (40, 3.61856 ) (60, 3.48485)};
        \addplot[fill=rust] coordinates {(0, 3.46667) (20, 3.08824 ) (40, 3.51948) (60, 3.24138)};
        \addplot[fill=silver] coordinates {(0, 3.69444) (20, 3.04) (40,3.66667) (60, 3.81081)};
        \legend{4-ring, 5-star, 6-ring, 7-star}
    \end{axis}
\end{tikzpicture}
\begin{tikzpicture}
    \begin{axis}[
        width=\columnwidth, %
        height=6cm,       %
        ybar,
        bar width = 6pt,
        xtick={0, 20, 40, 60},
        ytick={1,2,3,4,5,6},
        xticklabels={QAOA, VQE, QFT, RCA},
        ymin=0, ymax=6,
        ylabel={Improv. Factor (Required Lifetime)},
        legend style={at={(0.5,0.98)}, anchor=north, legend columns=-1},
        enlarge x limits=0.3,
        legend image code/.code={
        \draw[fill=#1] (0cm,-0.1cm) rectangle (0.4cm,0.1cm);}
    ]   
        \draw [dashed, thin, color=gray](-20,2) -- (100,2);
        \draw [dashed, thin, color=gray](-20,3) -- (100,3);
        \draw [dashed, thin, color=gray](-20,4) -- (100,4);
        \draw [dashed, thin, color=gray](-20,5) -- (100,5);
        \addplot[fill=customblue] coordinates {(0, 4.54286) (20, 4.58974) (40, 4.73) (60, 2.45902)};
        \addplot[fill=orange] coordinates {(0, 3.70) (20, 3.81579) (40, 4.26923) (60, 1.98305)};
        \addplot[fill=rust] coordinates {(0, 3.70) (20, 3.34483) (40, 3.84615 ) (60, 1.95)};
        \addplot[fill=silver] coordinates {(0, 4.88462) (20, 3.475) (40, 4.45) (60, 2.19298)};
        \legend{4-ring, 5-star, 6-ring, 7-star}
    \end{axis}
\end{tikzpicture}
\caption{Improvements of \name\ over baseline on 36-qubit QAOA, QFT, RCA, and VQE using 4 QPUs and different resource states as illustrated in \fig{rs}~(a). The improvement factor is defined as $
    f\equiv \tau^{\textbf{OneQ}}/\tau^{\textbf{DC-MBQC}}
$.}
\label{fig:improv}
\end{figure}

\noindent \textbf{Analysis of Performance Gains.} The substantial gains over the monolithic baseline stem from the parallelism inherent in the distributed architecture, which our framework enables by co-optimizing local computation and global synchronization. First, our adaptive graph partitioning (\alg{graph-partition}) optimizes modularity (while balancing workload) to produce structurally-optimized subgraphs. This reduces the complexity for the single-QPU compiler, allowing it to find a more compact local schedule (i.e., fewer execution layers) for each QPU. Second, our framework systematically reduces the synchronization overhead. The partitioning algorithm minimizes the number of inter-QPU connectors (cut edges) , while our BDIR layer scheduler (\alg{bdir}) heuristically reduces their scheduling cost (required photon lifetime).

\noindent \textbf{Different Resource States.}
Our framework also supports different resource states introduced in \sec{mbqc-assumptions}.
We compare the effect of 4 resource states on execution time improvements and required lifetime improvements in \fig{improv}.
The 4-ring RSG outperforms other candidates in most benchmarks.
For the required photon lifetime, the 6-ring RSG are consistently worse than other RSG type.
This is due to the unique topology. 
Supposing that a diagonal pair of photons in a 6-ring resource state is removed, there are two 2-qubit resource states left, both of which can be utilized for routing. Therefore, a 6-ring resource state can be used as routing resource twice, while other resource states can be used only once. Since the main challenge of compiling for a single QPU is to deal with complex routing demands, the 6-ring RSG is especially convenient for single QPU.

\noindent \textbf{Number of QPUs.} 
We further study the impact of using 8 QPUs for distributed MBQC in \tab{dcmbqc_results_8qpu}. 
The improvements are consistently better than the 4 QPUs case in \tab{dcmbqc_results_4qpu} with up to $7.46\times$ improvements on the required photon lifetime for 81-qubit QFT.

\subsection{Comprison with SOTA Compilers}\label{sec:sota}
As established in \sec{compiler}, the core task of monolithic compilers like OneQ, OnePerc, and OneAdapt is to solve the complex mapping of a computation graph onto a constrained 3D spatio-temporal resource grid. Our DC-MBQC framework is explicitly designed to be modular around this abstraction (see \fig{pipeline}). 
In this section, we clarify this relationship with OnePerc and present new comparative data against OneAdapt.

\noindent \textbf{OnePerc} OnePerc's primary contribution is its novel online, randomness-aware execution strategy (e.g., renormalization) to address probabilistic fusion failures~\cite{zhang2024oneperc}. This online problem is orthogonal to our framework. For its offline compilation pass (the stage that maps the logical graph to the 3D resource grid) OnePerc's authors state that they ``utilize" and ``extend" the foundational OneQ mapping algorithm. Therefore, our baseline comparison against OneQ in \sec{result} directly evaluates this core compilation stage, and the results are representative of the offline compilation performance for both OneQ and OnePerc.

\noindent \textbf{OneAdapt} 
OneAdapt  introduces a dynamic refreshing technique to mitigate photon loss. To adapt this monolithic compiler for our distributed comparison, we model the inter-QPU communication overhead via boundary resource reservation. Specifically, we dedicate the boundary resource states on each layer to serve as communication interfaces, effectively reducing the grid size by 2 in each dimension.
This approach provides a consistent estimation of the communication overhead, allowing us to focus on the algorithmic performance comparison between the distributed architecture and the single-QPU baseline.
Our experimental results, summarized in Table \tab{oneadapt}, demonstrate that our distributed approach still provides substantial additive gains on top of OneAdapt. For instance, with 8 QPUs, DC-MBQC achieves up to a $5.74\times$ execution time speedup and a $4.33\times$ reduction in required photon lifetime.

\begin{table*}[ht]
 \centering
 \caption{The results of \name\ and its relative performance to OneAdapt using 4 QPUs and 8 QPUs.}\label{tab:oneadapt}
 \begin{tabular}{c|c|c|c|c|c|c|c}
 \toprule
 \#QPUs & Program-\#Qubits & OneAdapt Exec. Time & Our Exec. Time & Improv. Factor & OneAdapt Lifetime & Our Lifetime & Improv. Factor \\
 \midrule
 \midrule
 \multirow{6}{*}{4} & VQE-64 & 306 & 92 & \textbf{3.32} & 13 & 6 & 2.17 \\
 & VQE-100 & 552 & 198 & 2.79 & 20 & 14 & 1.43 \\
 \cmidrule{2-8}
 & QAOA-64 & 332 & 111 & 2.99 & 14 & 7 & 2.00 \\
 & QAOA-121 & 705 & 246 & 2.87 & 20 & 13 & 1.54 \\
 \cmidrule{2-8}
 & QFT-36 & 313 & 101 & 3.10 & 10 & 9 & 1.11 \\
 & QFT-64 & 715 & 216 & 3.31 & 20 & 6 & \textbf{3.33} \\
 \cmidrule{1-8}
 \multirow{6}{*}{8} & VQE-64 & 327 & 57 & \textbf{5.74} & 20 & 6 & 3.33 \\ 
 & VQE-100 & 599 & 108 & 5.55 & 20 & 10 & 2.00 \\
 \cmidrule{2-8}
 & QAOA-64 & 326 & 71 & 4.59 & 9 & 3 & 3.00 \\
 & QAOA-121 & 706 & 175 & 4.03 & 20 & 5 & 4.00 \\
 \cmidrule{2-8}
 & QFT-36 & 319 & 76 & 4.20 & 13 & 3 & \textbf{4.33} \\
 & QFT-64 & 673 & 138 & 4.88 & 20 & 5 & 4.00 \\
 \bottomrule
 \end{tabular}
\end{table*}

\subsection{Component Analysis}
We compare our BDIR algorithm against a standard priority-based list scheduling baseline. In this experiment, we use the full DC-MBQC framework, only swapping the final layer scheduling component. \tab{bdir_results_4qpu} demonstrates that the BDIR algorithm consistently reduces the required photon lifetime compared to the baseline, achieving improvements ranging from 4.62\% to a significant 15.12\%.
\begin{table}[hb]
    \centering
    \caption{Effectiveness of BDIR}
    \label{tab:bdir_results_4qpu}
    \begin{tabular}{c|c|c|c}
        \toprule
        Program-\#Qubits & Baseline Lifetime & BDIR Lifetime & Improv. Factor \\
        \midrule
        \midrule
        QFT-16  &  33 &  31 &  6.06\% \\
        QFT-25  &  65 &  62 &  4.62\% \\
        QFT-36  &  86 &  73 & 15.12\% \\
        QFT-49  & 140 & 124 & 11.43\% \\
        QFT-64  & 196 & 184 &  6.12\% \\
        \bottomrule
    \end{tabular}
\end{table}

\subsection{Sensitivity Analysis}\label{sec:sensitivity}
In this section, we perform sensitivity analysis of two key parameters that are central to our framework: the hardware's connection capacity $K_{\max}$ and the maximum imbalance factor $\alpha_{\max}$ used in our adaptive graph partitioning algorithm (\alg{graph-partition}). We select 36-qubit QFT program as representative benchmark, as QFT requires the most number of 2-qubit gates and fusions in \tab{benchmark}. For $K_{\max}$, we also conduct sensitivity analysis on 25-qubit QFT instance to demonstrate the consistency of our findings.

For $K_{\max}$, \fig{sensitivity-K} exhibits a clear pattern of diminishing returns for both execution time and required lifetime when $K_{\max}$ increases, with the elbow point occuring at around $K_{\max} = 4$ to $7$. This trend suggests that inter-QPU communication is the primary bottleneck only for small $K_{\max}$, demonstrating that our \name\ framework achieves significant performance gains without demanding an impractically high connection capacity.

Parameter $\alpha_{\max}$ governs the trade-off between strict workload balance among QPUs and the structural quality of the resulting subgraphs in our adaptive graph partitioning algorithm. \fig{sensitivity-alpha} illustrates performance improvement of execution time and required lifetime is robust against $\alpha_{\max}$, fluctuating only within a narrow range. Remarkably, we find that for the entire tested range of $\alpha_{\max}$ (from 1.05 upwards), the graph partition produced by \alg{graph-partition} remains identical, with a constant cut size of 60 and a modularity of 0.74. Therefore, the observed performance variation does not stem from the partitioning quality but exclusively from the stochastic nature of subsequent local single-QPU compilation and global layer scheduling.

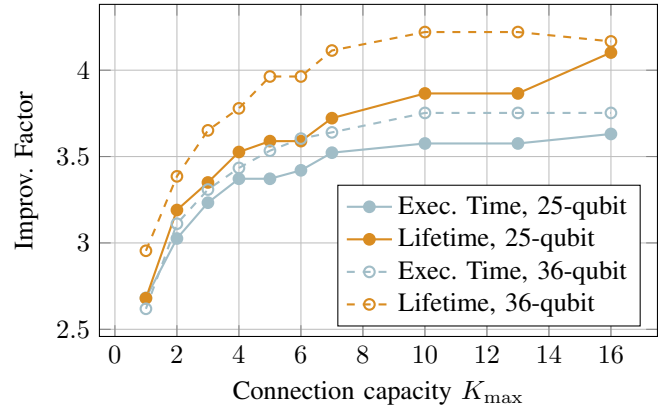
\begin{figure}[ht]
    \centering
    \begin{tikzpicture}
        \begin{axis}[
            xlabel={Connection capacity $K_{\max}$},
            ylabel={Improv. Factor},
            width=9cm,
            height=6cm,
            legend pos=south east,
            legend cell align={left},
            grid=major,
        ]
        \addplot[
            thick,
            solid,
            color=lightblue,
            mark=*,
            mark options={solid}
        ] coordinates {
            (1, 2.6818181818181817) (2, 3.025641025641025) (3, 3.232876712328767) (4, 3.3714285714285714) (5, 3.3714285714285714) (6, 3.420289855072464) (7, 3.5223880597014925) (10, 3.575757575757576) (13, 3.575757575757576) (16, 3.630769230769231) 
        };
        \addlegendentry{Exec. Time, 25-qubit}
        \addplot[
            thick,
            solid,
            color=orange,
            mark=*,
            mark options={solid}
        ] coordinates {
            (1, 2.68) (2, 3.1904761904761907) (3, 3.35) (4, 3.526315789473684) (5, 3.5892857142857144) (6, 3.5892857142857144) (7, 3.7222222222222223) (10, 3.865384615384616) (13, 3.865384615384616) (16, 4.1020408163265305) 
        };
        \addlegendentry{Lifetime, 25-qubit}

        \addplot[
            thick,
            dashed,
            color=lightblue,
            mark=o,
            mark options={solid}
        ] coordinates {
            (1, 2.618705035971223) (2, 3.111111111111111) (3, 3.309090909090909) (4, 3.4339622641509435) (5, 3.533980582524272) (6, 3.603960396039604) (7, 3.64) (10, 3.752577319587629) (13, 3.752577319587629) (16, 3.752577319587629) 
        };
        \addlegendentry{Exec. Time, 36-qubit}
        \addplot[
            thick,
            dashed,
            color=orange,
            mark=o,
            mark options={solid}
        ] coordinates {
            (1, 2.9545454545454546) (2, 3.3854166666666665) (3, 3.651685393258427) (4, 3.779069767441861) (5, 3.963414634146341) (6, 3.963414634146341) (7, 4.113924050632911) (10, 4.220779220779221) (13, 4.220779220779221) (16, 4.166666666666667) 
        };
        \addlegendentry{Lifetime, 36-qubit}

        \end{axis}
    \end{tikzpicture}
    \caption{Impact of connection capacity $K_{\max}$. }\label{fig:sensitivity-K}
\end{figure}

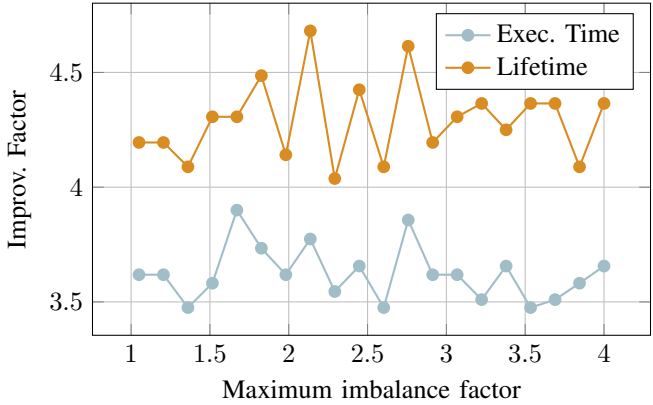
\begin{figure}[ht]
    \centering
    \begin{tikzpicture}
        \begin{axis}[
            xlabel={Maximum imbalance factor},
            ylabel={Improv. Factor},
            width=9cm,
            height=6cm,
            legend pos=north east,
            legend cell align={left},
            grid=major,
        ]
        \addplot[
            thick,
            color=lightblue,
            mark=*,
            mark options={solid}
        ] coordinates {
            (1.05, 3.618556701030928) (1.2052631578947368, 3.618556701030928) (1.3605263157894738, 3.4752475247524752) (1.5157894736842106, 3.5816326530612246) (1.6710526315789471, 3.9) (1.8263157894736843, 3.734042553191489) (1.981578947368421, 3.618556701030928) (2.136842105263158, 3.774193548387097) (2.292105263157895, 3.5454545454545454) (2.447368421052632, 3.65625) (2.602631578947369, 3.4752475247524752) (2.7578947368421054, 3.857142857142857) (2.913157894736842, 3.618556701030928) (3.068421052631579, 3.618556701030928) (3.223684210526316, 3.51) (3.378947368421053, 3.65625) (3.53421052631579, 3.4752475247524752) (3.689473684210527, 3.51) (3.844736842105263, 3.5816326530612246) (4.0, 3.65625)
        };
        \addlegendentry{Exec. Time}

        \addplot[
            thick,
            color=orange,
            mark=*,
            mark options={solid}
        ] coordinates {
            (1.05, 4.194805194805195) (1.2052631578947368, 4.194805194805195) (1.3605263157894738, 4.0886075949367084) (1.5157894736842106, 4.306666666666667) (1.6710526315789471, 4.306666666666667) (1.8263157894736843, 4.486111111111111) (1.981578947368421, 4.141025641025641) (2.136842105263158, 4.681159420289855) (2.292105263157895, 4.0375) (2.447368421052632, 4.424657534246576) (2.602631578947369, 4.0886075949367084) (2.7578947368421054, 4.614285714285714) (2.913157894736842, 4.194805194805195) (3.068421052631579, 4.306666666666667) (3.223684210526316, 4.364864864864865) (3.378947368421053, 4.25) (3.53421052631579, 4.364864864864865) (3.689473684210527, 4.364864864864865) (3.844736842105263, 4.0886075949367084) (4.0, 4.364864864864865)
        };
        \addlegendentry{Lifetime}

        \end{axis}
    \end{tikzpicture}
    \caption{Robustness of maximum imbalance factor $\alpha_{\max}$.}\label{fig:sensitivity-alpha}
\end{figure}

\subsection{Scalability} \label{sec:scalability}

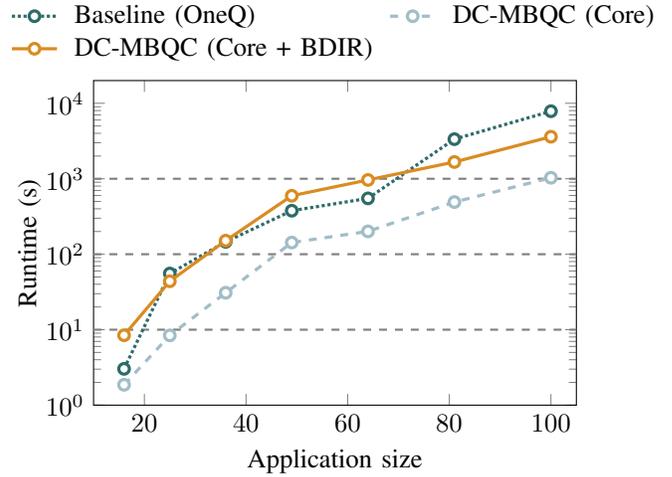
\begin{figure}[ht] 
\centering  
\begin{tikzpicture}
    \begin{axis}[
        xlabel={Application size},
        ylabel={Runtime (s)},
        ylabel style={yshift=-5pt},
        xmin=10,
        xmax=105,
        ymin=1e0, %
        ymax=2e4, %
        ymode=log,
        ytick={1e0, 1e1, 1e2, 1e3, 1e4},
        width=8cm,
        height=5.9cm,
        max space between ticks=40,
        legend style={draw=none, fill=none, at={(0.5,1.03)}, anchor=south, legend columns=2, column sep=4pt },
        clip mode=individual,
        legend cell align={left},
    ]
    
    \draw[gray, dashed, thick] (axis cs:10,10^3) -- (axis cs:105,10^3);
    \draw[gray, dashed, thick] (axis cs:10,10^1) -- (axis cs:105,10^1);
    \draw[gray, dashed, thick] (axis cs:10,10^2) -- (axis cs:105,10^2);

    \addplot[
        very thick,
        densely dotted,
        mark=*,
        mark size=2pt, %
        mark options={solid, fill=white},
        color=seagreen,
        opacity=1
    ] table [
        x=qubits,
        y=time,
        col sep=comma,
        header=true
    ] {data/runtime/oneq.txt};
    \addlegendentry{Baseline (OneQ)};

    \addplot[
        very thick,
        dashed,
        mark=*,
        mark size=2pt, %
        mark options={solid, fill=white},
        color=lightblue,
        opacity=1
    ] table [
        x=qubits,
        y=time,
        col sep=comma,
        header=true
    ] {data/runtime/dqc.txt};
    \addlegendentry{\name\ (Core)};

    \addplot[
        very thick,
        solid,
        mark=*,
        mark size=2pt, %
        mark options={solid, fill=white},
        color=orange,
        opacity=1
    ] table [
        x=qubits,
        y=time,
        col sep=comma,
        header=true
    ] {data/runtime/dqc-bdir.txt};
    \addlegendentry{\name\ (Core + BDIR)};
    
    \end{axis}
\end{tikzpicture}
\caption{Runtime scaling for compiling QFT programs (excluding common pre-processing overhead). \name\ demonstrates better scalability compared to Baseline. The efficiency of \name\ enables compilation of larger quantum programs within the same time budget.}
\label{fig:scaling}
\end{figure}

Compilation time is a key factor for executing large-scale applications~\cite{beverland2022assessing}. To evaluate the scalability of our compilation heuristics, \fig{scaling} compares the runtime of the core compilation stages (excluding common pre-processing) for QFT benchmarks up to 100 qubits. We compare the monolithic (OneQ) baseline against two versions of our 8-QPU framework: DC-MBQC (Core + BDIR) (using our full BDIR optimization with $I_{\max}=20$), and DC-MBQC (Core) (using only the basic priority-based list scheduling without BDIR). \fig{scaling} shows that while compilation time increases for all methods as problem size grows, our distributed approach demonstrates better scaling. The DC-MBQC (Core + BDIR) exhibits a less increase compared to baseline. Furthermore, the DC-MBQC (Core) version offers a compelling trade-off: by omitting the BDIR optimization pass, it achieves faster compilation times while still benefiting from the core framework's partitioning and basic scheduling advantages. This highlights the flexibility of our framework, allowing users to balance compilation effort against the final quality.

\section{Conclusions} \label{sec:conclusion}

In this work, we presented \name, the first distributed quantum compilation framework specifically designed for measurement-based quantum computing. 
Our framework addresses the challenges of distributed MBQC via an adaptive graph partitioning algorithm and a novel apporach to the layer scheduling problem. We introduce the metric of required photon lifetime to evaluate pratical performance. 
Experiments show a 7.46$\times$ reduction in required photon lifetime and a 6.82$\times$ computation speedup on 8 fully-connected QPUs, validating the effectiveness of our approach.

\section*{Acknowledgement}
We thank Zhixin Song for his insightful discussions and inspiring suggestions regarding the distributed architecture. We also gratefully acknowledge him for his substantial assistance with the literature review, figure preparation, and the drafting of the initial introduction. 
We thank the anonymous reviewers for their constructive comments that significantly improved the quality of this paper. 
This work was supported by the National Natural Science Foundation of China (Grant Number 92365117).


\begin{thebibliography}{00}
\bibitem{aghaee2025scaling}
H.~Aghaee~Rad, T.~Ainsworth, R.~N.~Alexander, B.~Altieri, M.~F.~Askarani, R.~Baby, L.~Banchi, B.~Q.~Baragiola, J.~E.~Bourassa, R.~S.~Chadwick, I.~Charania, H.~Chen, M.~J.~Collins, P.~Contu, N.~D'Arcy, G.~Dauphinais, R.~De~Prins, D.~Deschenes, I.~Di~Luch, S.~Duque, P.~Edke, S.~E.~Fayer, S.~Ferracin, H.~Ferretti, J.~Gefaell, S.~Glancy, C.~Gonz\'alez-Arciniegas, T.~Grainge, Z.~Han, J.~Hastrup, L.~G.~Helt, T.~Hillmann, J.~Hundal, S.~Izumi, T.~Jaeken, M.~Jonas, S.~Kocsis, I.~Krasnokutska, M.~V.~Larsen, P.~Laskowski, F.~Laudenbach, J.~Lavoie, M.~Li, E.~Lomonte, C.~E.~Lopetegui, B.~Luey, A.~P.~Lund, C.~Ma, L.~S.~Madsen, D.~H.~Mahler, L.~Mantilla~Calder\'on, M.~Menotti, F.~M.~Miatto, B.~Morrison, P.~J.~Nadkarni, T.~Nakamura, L.~Neuhaus, Z.~Niu, R.~Noro, K.~Papirov, A.~Pesah, D.~S.~Phillips, W.~N.~Plick, T.~Rogalsky, F.~Rortais, J.~Sabines-Chesterking, S.~Safavi-Bayat, E.~Sazhaev, M.~Seymour, K.~Rezaei~Shad, M.~Silverman, S.~A.~Srinivasan, M.~Stephan, Q.~Y.~Tang, J.~F.~Tasker, Y.~S.~Teo, R.~B.~Then, J.~E.~Tremblay, I.~Tzitrin, V.~D.~Vaidya, M.~Vasmer, Z.~Vernon, L.~F.~S.~S.~M.~Villalobos, B.~W.~Walshe, R.~Weil, X.~Xin, X.~Yan, Y.~Yao, M.~Zamani~Abnili, and Y.~Zhang, ``Scaling and networking a modular photonic quantum computer,'' \emph{Nature}, vol.~638, pp.~1--8, 2025.

\bibitem{ai2024quantum}
Google Quantum AI, ``Quantum error correction below the surface code threshold,'' \emph{Nature}, vol. 638, no. 8052, p. 920, 2024.

\bibitem{automated2019}
P.~Andr\'es-Mart\'{\i}nez and C.~Heunen, ``Automated distribution of quantum circuits via hypergraph partitioning,'' \emph{Phys. Rev. A}, vol. 100, p. 032308, Sep 2019.

\bibitem{timeslice2020}
J.~M. Baker, C.~Duckering, A.~Hoover, and F.~T. Chong, ``Time-sliced quantum circuit partitioning for modular architectures,'' in \emph{Proceedings of the 17th ACM International Conference on Computing Frontiers}, ser. CF '20.\hskip 1em plus 0.5em minus 0.4em\relax New York, NY, USA: Association for Computing Machinery, 2020, p. 98–107. 

\bibitem{Barral_2025}
D.~Barral, F.~J. Cardama, G.~Díaz-Camacho, D.~Faílde, I.~F. Llovo, M.~Mussa-Juane, J.~Vázquez-Pérez, J.~Villasuso, C.~Piñeiro, N.~Costas, J.~C. Pichel, T.~F. Pena, and A.~Gómez, ``Review of distributed quantum computing: From single {QPU} to high performance quantum computing,'' \emph{Computer Science Review}, vol.~57, p. 100747, Aug. 2025.

\bibitem{Barrett_2005}
S.~D. Barrett and P.~Kok, ``Efficient high-fidelity quantum computation using matter qubits and linear optics,'' \emph{Physical Review A}, vol.~71, no.~6, Jun. 2005.

\bibitem{bartolucci2023fusion}
S.~Bartolucci, P.~Birchall, H.~Bombin, H.~Cable, C.~Dawson, M.~Gimeno-Segovia, E.~Johnston, K.~Kieling, N.~Nickerson, M.~Pant, F.~Pastawski, T.~Rudolph, and C.~Sparrow, ``Fusion-based quantum computation,'' \emph{Nature Communications}, vol.~14, no.~1, p. 912, 2023.

\bibitem{beverland2022assessing}
M.~E. Beverland, P.~Murali, M.~Troyer, K.~M. Svore, T.~Hoefler, V.~Kliuchnikov, G.~H. Low, M.~Soeken, A.~Sundaram, and A.~Vaschillo, ``Assessing requirements to scale to practical quantum advantage,'' \emph{arXiv preprint arXiv:2211.07629}, 2022.

\bibitem{Blondel_2008}
V.~D. Blondel, J.-L. Guillaume, R.~Lambiotte, and E.~Lefebvre, ``Fast unfolding of communities in large networks,'' \emph{Journal of Statistical Mechanics: Theory and Experiment}, vol. 2008, no.~10, p. P10008, Oct. 2008.

\bibitem{bluvstein2024logical}
D.~Bluvstein, S.~J.~Evered, A.~A.~Geim, S.~H.~Li, H.~Zhou, T.~Manovitz, S.~Ebadi, M.~Cain, M.~Kalinowski, D.~Hangleiter, J.~P.~Bonilla~Ataides, N.~Maskara, I.~Cong, X.~Gao, P.~Sales~Rodriguez, T.~Karolyshyn, G.~Semeghini, M.~J.~Gullans, M.~Greiner, V.~Vuleti\'c, and M.~D.~Lukin, ``Logical quantum processor based on reconfigurable atom arrays,'' \emph{Nature}, vol.~626, no.~7997, pp.~58--65, Dec. 2023.


\bibitem{bombin2021interleaving}
H.~Bombin, I.~H. Kim, D.~Litinski, N.~Nickerson, M.~Pant, F.~Pastawski, S.~Roberts, and T.~Rudolph, ``Interleaving: Modular architectures for fault-tolerant photonic quantum computing,'' \emph{arXiv preprint arXiv:2103.08612}, 2021.

\bibitem{briegel2009measurement}
H.~J. Briegel, D.~E. Browne, W.~D{\"u}r, R.~Raussendorf, and M.~Van~den Nest, ``Measurement-based quantum computation,'' \emph{Nature Physics}, vol.~5, no.~1, pp. 19--26, 2009.

\bibitem{BROADBENT20092489}
A.~Broadbent and E.~Kashefi, ``Parallelizing quantum circuits,'' \emph{Theoretical Computer Science}, vol. 410, no.~26, pp. 2489--2510, 2009.

\bibitem{Browne_2005}
D.~E. Browne and T.~Rudolph, ``Resource-efficient linear optical quantum computation,'' \emph{Physical Review Letters}, vol.~95, no.~1, Jun. 2005.

\bibitem{caleffi2024distributed}
M.~Caleffi, M.~Amoretti, D.~Ferrari, J.~Illiano, A.~Manzalini, and A.~S. Cacciapuoti, ``Distributed quantum computing: a survey,'' \emph{Computer Networks}, vol. 254, p. 110672, 2024.


\bibitem{coppersmith2002approximate}
D.~Coppersmith, ``An approximate {F}ourier transform useful in quantum factoring,'' 2002, arXiv preprint arXiv:quant-ph/0201067.

\bibitem{covey2023quantum}
J.~P. Covey, H.~Weinfurter, and H.~Bernien, ``Quantum networks with neutral atom processing nodes,'' \emph{npj Quantum Information}, vol.~9, no.~1, p.~90, 2023.

\bibitem{cuccaro2004new}
S.~A. Cuccaro, T.~G. Draper, S.~A. Kutin, and D.~P. Moulton, ``A new quantum ripple-carry addition circuit,'' \emph{arXiv preprint quant-ph/0410184}, 2004.

\bibitem{dalzell2023quantum}
A.~M. Dalzell, S.~McArdle, M.~Berta, P.~Bienias, C.-F. Chen, A.~Gily{\'e}n, C.~T. Hann, M.~J. Kastoryano, E.~T. Khabiboulline, A.~Kubica,  G.~Salton, S.~Wang, and F.~G.~S.~L.~Brand{\~a}o, ``Quantum algorithms: {A} survey of applications and end-to-end complexities,'' \emph{arXiv preprint arXiv:2310.03011}, 2023.



\bibitem{eltes2020integrated}
F.~Eltes, G.~E. Villarreal-Garcia, D.~Caimi, H.~Siegwart, A.~A. Gentile, A.~Hart, P.~Stark, G.~D. Marshall, M.~G. Thompson, J.~Barreto, J.~Fompeyrine, and S.~Abel, ``An integrated optical modulator operating at cryogenic temperatures,'' \emph{Nature Materials}, vol.~19, no.~11, pp. 1164--1168, 2020.

\bibitem{farhi2014quantum}
E.~Farhi, J.~Goldstone, and S.~Gutmann, ``A quantum approximate optimization algorithm,'' 2014, arXiv preprint arXiv:1411.4028.

\bibitem{Ferrari_2021}
D.~Ferrari, A.~S. Cacciapuoti, M.~Amoretti, and M.~Caleffi, ``Compiler design for distributed quantum computing,'' \emph{IEEE Transactions on Quantum Engineering}, vol.~2, p. 1–20, 2021.

\bibitem{gao2025establishing}
D.~Gao, D.~Fan, C.~Zha, J.~Bei, G.~Cai, J.~Cai, S.~Cao, X.~Zeng, F.~Chen, J.~Chen, K.~Chen, X.~Chen, X.~Chen, Z.~Chen, Z.~Chen, Z.~Chen, W.~Chu, H.~Deng, Z.~Deng, P.~Ding, X.~Ding, Z.~Ding, S.~Dong, Y.~Dong, B.~Fan, Y.~Fu, S.~Gao, L.~Ge, M.~Gong, J.~Gui, C.~Guo, S.~Guo, X.~Guo, T.~He, L.~Hong, Y.~Hu, H.-L.~Huang, Y.-H.~Huo, T.~Jiang, Z.~Jiang, H.~Jin, Y.~Leng, D.~Li, D.~Li, F.~Li, J.~Li, J.~Li, J.~Li, J.~Li, N.~Li, S.~Li, W.~Li, Y.~Li, Y.~Li, F.~Liang, X.~Liang, N.~Liao, J.~Lin, W.~Lin, D.~Liu, H.~Liu, M.~Liu, X.~Liu, X.~Liu, Y.~Liu, H.~Lou, Y.~Ma, L.~Meng, H.~Mou, K.~Nan, B.~Nie, M.~Nie, J.~Ning, L.~Niu, W.~Peng, H.~Qian, H.~Rong, T.~Rong, H.~Shen, Q.~Shen, H.~Su, F.~Su, C.~Sun, L.~Sun, T.~Sun, Y.~Sun, Y.~Tan, J.~Tan, L.~Tang, W.~Tu, C.~Wan, J.~Wang, B.~Wang, C.~Wang, C.~Wang, C.~Wang, J.~Wang, L.~Wang, R.~Wang, S.~Wang, X.~Wang, Z.~Wei, J.~Wei, D.~Wu, G.~Wu, J.~Wu, S.~Wu, Y.~Wu, S.~Xie, L.~Xin, Y.~Xu, C.~Xue, K.~Yan, W.~Yang, X.~Yang, Y.~Yang, Y.~Ye, Z.~Ye, C.~Ying, J.~Yu, Q.~Yu, W.~Yu, S.~Zhan, F.~Zhang, H.~Zhang, K.~Zhang, P.~Zhang, W.~Zhang, Y.~Zhang, Y.~Zhang, L.~Zhang, G.~Zhao, P.~Zhao, X.~Zhao, X.~Zhao, Y.~Zhao, Z.~Zhao, L.~Zheng, F.~Zhou, L.~Zhou, N.~Zhou, N.~Zhou, S.~Zhou, S.~Zhou, Z.~Zhou, C.~Zhu, Q.~Zhu, G.~Zou, H.~Zou, Q.~Zhang, C.-Y.~Lu, C.-Z.~Peng, X.~Zhu, and J.-W.~Pan, ``Establishing a new benchmark in quantum computational advantage with 105-qubit {Z}uchongzhi 3.0 processor,''\emph{Physical Review Letters}, vol. 134, no.~9, p. 090601, 2025.

\bibitem{Garey1979JohnsonCA}
M.~R. Garey, in \emph{Johnson: computers and intractability: a guide to the theory of {NP}-completeness}, 1979.

\bibitem{gidney2021factor}
C.~Gidney and M.~Eker{\aa}, ``How to factor 2048 bit rsa integers in 8 hours using 20 million noisy qubits,'' \emph{Quantum}, vol.~5, p. 433, 2021.

\bibitem{grover1996fast}
L.~K. Grover, ``A fast quantum mechanical algorithm for database search,'' in \emph{Proceedings of the Twenty-Eighth Annual ACM Symposium on the Theory of Computing}, 1996, pp. 212--219.

\bibitem{guo2024boosted}
Y.-P. Guo, G.-Y. Zou, X.~Ding, Q.-H. Zhang, M.-C. Xu, R.-Z. Liu, J.-Y. Zhao, Z.-X. Ge, L.-C. Peng, K.-M. Xu, Y.-Y.~Lou, Z.~Ning, L.-J.~Wang, H.~Wang, Y.-H.~Huo, Y.-M.~He, C.-Y.~Lu, and J.-W.~Pan, ``Boosted fusion gates above the percolation threshold for scalable graph-state generation,'' \emph{arXiv preprint arXiv:2412.18882}, 2024.


\bibitem{hein2004multiparty}
M.~Hein, J.~Eisert, and H.~J. Briegel, ``Multiparty entanglement in graph states,'' \emph{Physical Review A—Atomic, Molecular, and Optical Physics}, vol.~69, no.~6, p. 062311, 2004.

\bibitem{Hucul_2014}
D.~Hucul, I.~V. Inlek, G.~Vittorini, C.~Crocker, S.~Debnath, S.~M. Clark, and C.~Monroe, ``Modular entanglement of atomic qubits using photons and phonons,'' \emph{Nature Physics}, vol.~11, no.~1, p. 37–42, Nov. 2014.

\bibitem{jing2025circuit}
M.~Jing, C.~Zhu, and X.~Wang, ``Circuit knitting faces exponential sampling overhead scaling bounded by entanglement cost,'' \emph{Physical Review A}, vol. 111, no.~1, p. 012433, 2025.


\bibitem{kandala2017hardware}
A.~Kandala, A.~Mezzacapo, K.~Temme, M.~Takita, M.~Brink, J.~M. Chow, and J.~M. Gambetta, ``Hardware-efficient variational quantum eigensolver for small molecules and quantum magnets,'' \emph{Nature}, vol. 549, no. 7671, pp. 242--246, 2017.

\bibitem{KARYPIS199896}
G.~Karypis and V.~Kumar, ``Multilevelk-way partitioning scheme for irregular graphs,'' \emph{Journal of Parallel and Distributed Computing}, vol.~48, no.~1, pp. 96--129, 1998.

\bibitem{leung2019deterministic}
N.~Leung, Y.~Lu, S.~Chakram, R.~Naik, N.~Earnest, R.~Ma, K.~Jacobs, A.~Cleland, and D.~Schuster, ``Deterministic bidirectional communication and remote entanglement generation between superconducting qubits,'' \emph{npj quantum information}, vol.~5, no.~1, p.~18, 2019.

\bibitem{li2024high}
Y.~Li and J.~D. Thompson, ``High-rate and high-fidelity modular interconnects between neutral atom quantum processors,'' \emph{PRX Quantum}, vol.~5, no.~2, p. 020363, 2024.

\bibitem{loschnauer2024scalable}
C.~L{\"o}schnauer, J.~M. Toba, A.~Hughes, S.~King, M.~Weber, R.~Srinivas, R.~Matt, R.~Nourshargh, D.~Allcock, C.~Ballance, C.~Matthiesen, M.~Malinowski, and T.~P.~Harty, ``Scalable, high-fidelity all-electronic control of trapped-ion qubits,'' \emph{arXiv preprint arXiv:2407.07694}, 2024.

\bibitem{main2025distributed}
D.~Main, P.~Drmota, D.~Nadlinger, E.~Ainley, A.~Agrawal, B.~Nichol, R.~Srinivas, G.~Araneda, and D.~Lucas, ``Distributed quantum computing across an optical network link,'' \emph{Nature}, pp. 1--6, 2025.

\bibitem{mo2024fcm}
Z.~Mo, Y.~Li, A.~Pawar, X.~Tang, J.~Yang, and Y.~Zhang, ``Fcm: A fusion-aware wire cutting approach for measurement-based quantum computing,'' in \emph{Proceedings of the 61st ACM/IEEE Design Automation Conference}, 2024, pp. 1--6.

\bibitem{moses2023race}
S.~A.~Moses, C.~H.~Baldwin, M.~S.~Allman, R.~Ancona, L.~Ascarrunz, C.~Barnes, 
J.~Bartolotta, B.~Bjork, P.~Blanchard, M.~Bohn, J.~G.~Bohnet, N.~C.~Brown, 
N.~Q.~Burdick, W.~C.~Burton, S.~L.~Campbell, J.~P.~Campora, C.~Carron, 
J.~Chambers, J.~W.~Chan, Y.~H.~Chen, A.~Chernoguzov, E.~Chertkov, J.~Colina, 
J.~P.~Curtis, R.~Daniel, M.~DeCross, D.~Deen, C.~Delaney, J.~M.~Dreiling, 
C.~T.~Ertsgaard, J.~Esposito, B.~Estey, M.~Fabrikant, C.~Figgatt, C.~Foltz, 
M.~Foss-Feig, D.~Francois, J.~P.~Gaebler, T.~M.~Gatterman, C.~N.~Gilbreth, 
J.~Giles, E.~Glynn, A.~Hall, A.~M.~Hankin, A.~Hansen, D.~Hayes, B.~Higashi, 
I.~M.~Hoffman, B.~Horning, J.~J.~Hout, R.~Jacobs, J.~Johansen, L.~Jones, 
J.~Karcz, T.~Klein, P.~Lauria, P.~Lee, D.~Liefer, S.~T.~Lu, D.~Lucchetti, 
C.~Lytle, A.~Malm, M.~Matheny, B.~Mathewson, K.~Mayer, D.~B.~Miller, 
M.~Mills, B.~Neyenhuis, L.~Nugent, S.~Olson, J.~Parks, G.~N.~Price, Z.~Price, 
M.~Pugh, A.~Ransford, A.~P.~Reed, C.~Roman, M.~Rowe, C.~Ryan-Anderson, 
S.~Sanders, J.~Sedlacek, P.~Shevchuk, P.~Siegfried, T.~Skripka, B.~Spaun, 
R.~T.~Sprenkle, R.~P.~Stutz, M.~Swallows, R.~I.~Tobey, A.~Tran, T.~Tran, 
E.~Vogt, C.~Volin, J.~Walker, A.~M.~Zolot, and J.~M.~Pino, ``A race-track trapped-ion quantum processor,'' \emph{Physical Review X}, vol.~13, no.~4, p. 041052, 2023.


\bibitem{PhysRevE.69.066133}
M.~E.~J. Newman, ``Fast algorithm for detecting community structure in networks,'' \emph{Phys. Rev. E}, vol.~69, p. 066133, Jun 2004.

\bibitem{Nielsen_2004}
M.~A. Nielsen, ``Optical quantum computation using cluster states,'' \emph{Physical Review Letters}, vol.~93, no.~4, Jul. 2004.

\bibitem{paesani2019generation}
S.~Paesani, Y.~Ding, R.~Santagati, L.~Chakhmakhchyan, C.~Vigliar, K.~Rottwitt, L.~K. Oxenl{\o}we, J.~Wang, M.~G. Thompson, and A.~Laing, ``Generation and sampling of quantum states of light in a silicon chip,'' \emph{Nature Physics}, vol.~15, no.~9, pp. 925--929, 2019.

\bibitem{Papadimitiou_1976}
C.~H. Papadimitriou, ``The np-completeness of the bandwidth minimization problem,'' \emph{Computing}, vol.~16, no.~3, pp. 263--270, 1976.

\bibitem{Peng_2020}
T.~Peng, A.~W. Harrow, M.~Ozols, and X.~Wu, ``Simulating large quantum circuits on a small quantum computer,'' \emph{Physical Review Letters}, vol. 125, no.~15, Oct. 2020.

\bibitem{peruzzo2014variational}
A.~Peruzzo, J.~McClean, P.~Shadbolt, M.-H. Yung, X.-Q. Zhou, P.~J. Love, A.~Aspuru-Guzik, and J.~L. O’brien, ``A variational eigenvalue solver on a photonic quantum processor,'' \emph{Nature Communications}, vol.~5, no.~1, p. 4213, 2014.

\bibitem{Pinedo}
M.~Pinedo, \emph{Scheduling : theory, algorithms, and systems}, 6th~ed.\hskip 1em plus 0.5em minus 0.4em\relax Cham: Springer, 2022.


\bibitem{preskill2025beyond}
J.~Preskill, ``Beyond NISQ: The megaquop machine,'' \emph{arXiv preprint arXiv:2502.17368}, 2025.

\bibitem{psiquantum2025manufacturable}
PsiQuantum, ``A manufacturable platform for photonic quantum computing,'' \emph{Nature}, pp. 1--3, 2025.

\bibitem{raussendorf2001one}
R.~Raussendorf and H.~J. Briegel, ``A one-way quantum computer,'' \emph{Physical Review Letters}, vol.~86, no.~22, p. 5188, 2001.

\bibitem{raussendorf2003measurement}
R.~Raussendorf, D.~E. Browne, and H.~J. Briegel, ``Measurement-based quantum computation on cluster states,'' \emph{Physical Review A}, vol.~68, no.~2, p. 022312, 2003.

\bibitem{ritter2012elementary}
S.~Ritter, C.~N{\"o}lleke, C.~Hahn, A.~Reiserer, A.~Neuzner, M.~Uphoff, M.~M{\"u}cke, E.~Figueroa, J.~Bochmann, and G.~Rempe, ``An elementary quantum network of single atoms in optical cavities,'' \emph{Nature}, vol. 484, no. 7393, pp. 195--200, 2012.

\bibitem{shor1999polynomial}
P.~W. Shor, ``Polynomial-time algorithms for prime factorization and discrete logarithms on a quantum computer,'' \emph{SIAM Review}, vol.~41, no.~2, pp. 303--332, 1999.

\bibitem{sinclair2024fault}
J.~Sinclair, J.~Ramette, B.~Grinkemeyer, D.~Bluvstein, M.~Lukin, and V.~Vuleti{\'c}, ``Fault-tolerant optical interconnects for neutral-atom arrays,'' \emph{arXiv preprint arXiv:2408.08955}, 2024.

\bibitem{stephenson2020high}
L.~Stephenson, D.~Nadlinger, B.~Nichol, S.~An, P.~Drmota, T.~Ballance, K.~Thirumalai, J.~Goodwin, D.~Lucas, and C.~Ballance, ``High-rate, high-fidelity entanglement of qubits across an elementary quantum network,'' \emph{Physical Review Letters}, vol. 124, no.~11, p. 110501, 2020.

\bibitem{stockill2017phase}
R.~Stockill, M.~Stanley, L.~Huthmacher, E.~Clarke, M.~Hugues, A.~Miller, C.~Matthiesen, C.~Le~Gall, and M.~Atat{\"u}re, ``Phase-tuned entangled state generation between distant spin qubits,'' \emph{Physical Review Letters}, vol. 119, no.~1, p. 010503, 2017.

\bibitem{Tang_2021}
W.~Tang, T.~Tomesh, M.~Suchara, J.~Larson, and M.~Martonosi, ``Cutqc: using small quantum computers for large quantum circuit evaluations,'' in \emph{Proceedings of the 26th ACM International Conference on Architectural Support for Programming Languages and Operating Systems}, ser. ASPLOS ’21.\hskip 1em plus 0.5em minus 0.4em\relax ACM, Apr. 2021, p. 473–486.

\bibitem{Traag_2019}
V.~A. Traag, L.~Waltman, and N.~J. van Eck, ``From louvain to leiden: guaranteeing well-connected communities,'' \emph{Scientific Reports}, vol.~9, no.~1, Mar. 2019. 

\bibitem{743431}
W.~Unger, ``The complexity of the approximation of the bandwidth problem,'' in \emph{Proceedings 39th Annual Symposium on Foundations of Computer Science (Cat. No.98CB36280)}, 1998, pp. 82--91.

\bibitem{young2022architecture}
C.~Young, A.~Safari, P.~Huft, J.~Zhang, E.~Oh, R.~Chinnarasu, and M.~Saffman, ``An architecture for quantum networking of neutral atom processors,'' \emph{Applied Physics B}, vol. 128, no.~8, p. 151, 2022.

\bibitem{zhang2024oneperc}
H.~Zhang, J.~Ruan, H.~Shapourian, R.~R. Kompella, and Y.~Ding, ``OnePerc: A randomness-aware compiler for photonic quantum computing,'' in \emph{Proceedings of the 29th ACM International Conference on Architectural Support for Programming Languages and Operating Systems, Volume 3}, 2024, pp. 738--754.

\bibitem{zhang2025oneadaptadaptivecompilationresourceconstrained}
H.~Zhang, J.~Ruan, D.~Tullsen, Y.~Ding, A.~Li, and T.~S. Humble, ``OneAdapt: Adaptive compilation for resource-constrained photonic one-way quantum computing,'' 2025. 

\bibitem{zhang2023oneq}
H.~Zhang, A.~Wu, Y.~Wang, G.~Li, H.~Shapourian, A.~Shabani, and Y.~Ding, ``OneQ: A compilation framework for photonic one-way quantum computation,'' in \emph{Proceedings of the 50th Annual International Symposium on Computer Architecture}, 2023, pp. 1--14.


\end{thebibliography}
\end{document}